\documentclass[a4paper,UKenglish,cleveref, autoref, thm-restate]{lipics-v2021}
\usepackage[ruled,vlined]{algorithm2e} % for algorithm-blocks.

\newcommand{\etal}{et al.}

\newcommand{\floor}[1]{\lfloor #1 \rfloor}

\title{Dynamic Binary Search Trees: Improved Lower Bounds for the Greedy-Future Algorithm}

\titlerunning{Dynamic BSTs: Improved Lower Bounds for Greedy-Future} % Optional, to be used if the title is long.

\author{Yaniv {Sadeh}}{Tel Aviv University, Israel}{yanivsadeh@mail.tau.ac.il}{https://orcid.org/0000-0002-5712-1028}{} %Please use full name; only 1 author per \author macro; first two parameters are mandatory, other parameters can be empty. Please provide at least the name of the affiliation and the country. The full address is optional. Use additional curly braces to indicate the correct name splitting when the last name consists of multiple name parts.
\author{Haim {Kaplan}}{Tel Aviv University, Israel}{haimk@tau.ac.il}{https://orcid.org/0000-0001-9586-8002}{}

\authorrunning{Y. Sadeh and H. Kaplan} % Mandatory. First: Use abbreviated first/middle names. Second (only in severe cases): Use first author plus 'et al.'

\Copyright{Yaniv Sadeh and Haim Kaplan} % Mandatory. Please use full first names. LIPIcs license is "CC-BY";  http://creativecommons.org/licenses/by/3.0/

\ccsdesc[300]{Theory of computation~Online algorithms}
\ccsdesc[500]{Theory of computation~Sorting and searching}

\keywords{Binary Search Trees, Greedy Future, Geometric Greedy, Lower Bounds, Dynamic Optimality Conjecture} % Mandatory; please add comma-separated list of keywords.

\funding{The work of the authors is partially supported by Israel Science Foundation (ISF) grant number 1595-19, German Science Foundation (GIF) grant number 1367 and the Blavatnik research fund at Tel Aviv University.} %optional, to capture a funding statement, which applies to all authors. Please enter author specific funding statements as fifth argument of the \author macro.

\nolinenumbers %uncomment to disable line numbering

\EventEditors{}
\EventNoEds{4}
\EventLongTitle{Updated version of the paper published at STACS 2023}
\EventShortTitle{}
\EventAcronym{STACS}
\EventYear{2023}
\EventDate{}
\EventLocation{}
\EventLogo{}
\SeriesVolume{0}
\ArticleNo{53}
\begin{document}

\maketitle

\begin{abstract}
Binary search trees (BSTs) are one of the most basic and widely used data structures. The best static tree for serving a sequence of queries (searches) can be computed by dynamic programming. In contrast, when the BSTs are allowed to be dynamic (i.e.\ change by rotations between searches), we still do not know how to compute the optimal algorithm (OPT) for a given sequence. One of the candidate algorithms whose serving cost is suspected to be optimal up-to a (multiplicative) constant factor is known by the name Greedy Future (GF). In an equivalent geometric way of representing queries on BSTs, GF is in fact equivalent to another algorithm called Geometric Greedy (GG). Most of the results on GF  are obtained using the geometric model and the study of GG. Despite this intensive recent fruitful research, the best lower bound we have on the competitive ratio of GF is $\frac{4}{3}$. Furthermore, it has been conjectured that the additive gap between the cost of GF and OPT is only linear in the number of queries. In this paper we prove a lower bound of $2$ on the competitive ratio of GF, and we prove that the additive gap between the cost of GF and OPT can be $\Omega(m \cdot \log\log n)$ where $n$ is the number of items in the tree and $m$ is the number of queries.
\end{abstract}

\pagebreak % Per instruction: first page contains only title, abstract, and authors' info.

%%% %%% %%% \input{section_introduction.tex}
\section{Introduction}
\label{section_introduction}

Binary search trees (BSTs) are one of the most basic and widely used data-structures. They are used to store a sorted set of keys
from a totally ordered universe. Traversing BSTs is usually done by using a single pointer, initially pointing to the root, and moving to the left or right child according to the order of the searched key and the key of the item at the current node. Therefore, we typically define the cost\footnote{Our cost model is formally defined in Definition~\ref{definition_avg_cost}, in Section~\ref{section_model}.} of a search to be the length of the search path. The data structure itself may be static, or change dynamically throughout time, in response to insertions and deletions of items, and possibly even restructured during queries.

Static BSTs are well understood. One can guarantee that the longest path from the root to a leaf is of length  $O(\log n)$ if the number of keys is $n$, by using a balanced tree. If the access sequence is known in advance (in fact only the frequency of accesses of each key matters) then an $O(n^2)$ time algorithm computing the optimal static tree for the particular set of frequencies was given by Knuth \cite{KnuthOptStatic1971}. It is also notable that the lower bound on the cost when the known frequencies are $\vec{f} = [f_1,f_2,\ldots,f_n]$ and the number of queries is $m$, is $\Omega(m \cdot H(\vec{f}))$ where $H(\vec{f}) = \sum_{i=1}^{n}{f_i \log \frac{1}{f_i}}$ is the entropy function. A simple way with $O(n \log n)$ running time to construct a near-optimal static (centroid) tree whose cost is $O(m \cdot (H(\vec{f})+1))$, has been described by Mehlhorn~\cite{StaticNearOptimalTreeMehlhorn}. The running time has been improved to $O(n)$ by Fredman~\cite{Fredman1975LinearAlmostOptBST}.

In contrast to the static case, the dynamic case is less understood. One can, of course, serve the sequence with a static tree. But, for many sequences we must change the structure of the tree as we make the searches in order to be efficient. For example, the requested items may be different in different parts of the sequence so a different set of items has to be placed near the root during different parts of the sequence. Restructuring is done by rotations that maintain the symmetric order. When rotations are allowed, the cost is defined to be the size of the subtree that contains the search path and all edges which we rotate.
 
Here,  we assume that the set of values stored in the tree does not change (no insertions or deletions), yet restructuring the tree is allowed to speed up future searches. One famous dynamic algorithm for doing this is the \emph{Splay} algorithm of Sleator and Tarjan~\cite{Splay1985}. After each query, the splay algorithm moves the queried item to the root of the tree, according to three simple rules called \emph{zig-zag}, \emph{zig-zig} and \emph{zig}.  The  splay algorithm is efficient in the sense that it is able to exploit the structure of many families of sequences. In particular splay is proven to be as good as the static optimum (up to a constant factor), which also implies that the cost of splay on any given sequence is at most $O(\log n)$ times the (dynamic) optimum cost. Sleator and Tarjan conjectured that splay is in fact dynamically-optimal, meaning that its  cost is like the cost of an optimal algorithm that knows the whole sequence of queries in advance, up to some constant factor. However, this dynamic-optimality conjecture  of splay is still open. In fact, it is open whether there is any dynamically-optimal online binary search tree algorithm. The best competitive ratio achievable to date is $O(\log \log n)$, and it is obtained by Tango~\cite{TangoTrees}, Multi-splay~\cite{MultiSplay} and Chain-splay~\cite{ChainSplay} trees, and a geometric divide-and-conquer approach of \cite{GeometricLgLgN_Chalermsook}.

While seeking for (better) guaranteed competitiveness, other dynamic algorithms were considered. A promising candidate was independently proposed by Lucas~\cite{LucasGF1988} and Munro~\cite{MunroGF2000}, which is now commonly referred to as \emph{Greedy Future}, henceforth: $GF$ in short. As its name suggests, $GF$ is a greedy algorithm that rearranges the nodes on the path from the root to the current queried item as a treap whose priorities are according to the future accesses\footnote{Each item in a treap has two keys:  {\em value} and {\em priority}. The treap is a binary search tree with respect to the values of the items and a heap with respect to their priorities. That is, the priority of an item is no larger than the priorities of its children. In our case, the priorities are deterministically defined by future requests in a way that we define precisely in Algorithm~\ref{alg_greedy_future}.} (as this paper deals with analyzing $GF$, we detail it formally in Algorithm~\ref{alg_greedy_future}). Note that unlike splay, $GF$, by definition, is required to know the future in order to restructure the tree. Surprisingly however, Demaine \etal~\cite{TheGeometryOfBSTs-SODA2009} showed that one can simulate $GF$ without knowing the future by a hierarchy of split-trees while losing only a constant factor in performance.
 
Additionally, \cite{TheGeometryOfBSTs-SODA2009} presented a  geometric view of an algorithm serving queries by a dynamic binary search tree using a two-dimensional grid on which we mark the sequence as well as the items accessed by the algorithm. In this presentation there is yet another natural promising candidate for dynamic optimality, which is commonly known as \emph{Geometric Greedy} and sometimes simply \emph{Greedy}, which we shall refer to as $GG$. \cite{TheGeometryOfBSTs-SODA2009} showed that  $GG$ is in fact the same algorithm as $GF$.

The geometric view proved useful to  obtain new results regarding $GG$ and hence $GF$. Fox~\cite{Fox2011GGAccessLemma} proved that an \emph{access-lemma} that is analogous to the so-called \emph{access-lemma} of splay trees holds for $GG$. From this follows that most of the nice properties that hold for splay also hold for $GF$. In particular, it follows  that $GF$ is $O(\log n)$ competitive. Chalermsook \etal~\cite{PatternAvoiding2015} analyzed upper bounds on the cost of $GG$ for access patterns which are permutations, and in particular found that for highly structured permutations, which they called \emph{$k$-decomposable}, the cost is $n \cdot 2^{\alpha(n)^{O(k)}}$ where $\alpha(n)$ is the inverse-Ackermann function. Chalermsook \etal~\cite{ImprovedPatternAvoidanceBounds-SODA2023} study special access patterns that belong to a broader family of \emph{pattern-avoiding} permutations. See \cite{TheLandscapeOfBSTs2016} for a survey of currently known properties of greedy and splay.

\medskip

\noindent
\textbf{Our Contributions:}
\begin{enumerate}
    \item It is known that $GF$ is not exactly optimal, but it  is conjectured, like splay, to be optimal up to a constant factor. In fact, it has been even more strongly conjectured by Demaine~\etal~\cite{TheGeometryOfBSTs-SODA2009} to be optimal up to an additive $O(m)$ term, and possibly even exactly $m$. Kozma~\cite{KozmaThesis} refuted the second part and gave a specific sequence for which this additive gap is $m+1$. In this paper we refute the linear gap conjecture and show a family of sequences for which the additive gap is at least $\Omega(m \log \log n)$. 
    
    \item The largest lower bound on the competitive ratio of $GF$ is $\frac{4}{3}$ by Demaine~\etal~\cite{TheGeometryOfBSTs-SODA2009}.
    They show a family of sequences on which after an initial query, the optimum pays $1.5$ on average per query while $GF$ pays $2$.\footnote{Reddmann~\cite{reddmann2021geometric} found an example in which the cost ratio between $GF$ and the optimum is $\frac{26}{17} \approx 1.53$. But this is for one particular sequence of a fixed length so it does not rule out any competitive ratio if we allow an additive constant.}
    We describe a technique that allows us to improve this lower bound to $2$.
    We note that  the best known lower bound on the competitive ratio of splay is $2$ (see \cite[Section 2.5]{LevyTarjan2019Thesis}). In both cases, the construction requires a rather large number of items (large $n$).
    
    \item \label{contribution_mult} Based on the multiplicative lower bound described above we show the following two interesting properties of $GF$: (1) There are sequences $X$ such that the cost of $GF$ on the reverse sequence is arbitrarily close to twice as large as the cost of $GF$ on $X$. (2) There are sequences $X$ such that  we can remove some queries from them and get a subsequence $X'$, such that the cost of $GF$ on $X'$ is arbitrarily close to twice as large as the cost of $GF$ on $X$.

\end{enumerate}

We study subsequences and reversal (contribution \ref{contribution_mult}) since any dynamically-optimal algorithm $A$ must have a ``nice'' behavior in these cases. Concretely, $A$ must satisfy the approximately-monotone property (Definition~\ref{definition_approximate_monotone}) which states that there is a fixed constant $c$ such that the cost of $A$ on any subsequence of any sequence is never more than $c$ times the cost on the whole sequence. As for reversal, the optimum can process a sequence and its reversal with similar costs up to a difference of $n$, thus any dynamically-optimal algorithm must be able to do so with costs that differ by at most a constant factor. We discuss this motivation in more  detail in Section~\ref{section_stable_sequences} (right after stating Theorem~\ref{theorem_GF_on_reverse_sequence}).

Our contributions are all based on the same technique, which is quite simple. We enforce $GF$ to maintain a static tree and only query the leaves of this tree. Although being dynamic in general, there are some access-patterns that cause $GF$ not to change the tree. By studying these patterns, we can study $GF$ on a static tree, and the analysis of its cost simplifies to the weighted-average of the depth of the queries (weighted by frequency). To lower-bound the gap between $GF$ and $OPT$, we analyze the average cost that can be saved by promoting the items in the leaves to locations closer to the root. Note that any other item can be placed further away from the root since it is never queried by the sequence.

%%% %%% %%% \input{section_model.tex}
\section{Model}
\label{section_model}

In this section we describe the model which we use, and define our notations. First, we note that throughout the paper $\lg x$ is used to denote the base two logarithm of $x$.

We consider a totally ordered universe of (fixed size) $n$ items. For simplicity, one may think of the values $\{1,\ldots,n\}$. The items are organized in some initial BST which we denote by $T_0$. Then, a sequence of queries, denoted by $X = [x_1,x_2,\ldots,x_m]$, is given, one query at a time. We reserve $m$ to denote the length of the sequence. The tree before serving $x_t$ is denoted by $T_{t-1}$. An algorithm has to find the queried value $x_t$, by traversing $T_{t-1}$ from its root. After finding $x_t$, the algorithm is allowed to re-structure $T_{t-1}$ to get $T_{t}$. We define the cost of the algorithm at time $t$ to be the total number of nodes that were touched at time $t$, both on the path to $x_t$ and for restructuring. The cost of an algorithm for the whole sequence is simply the sum of its costs over all times. We define it formally below.

\begin{definition}[Cost]
\label{definition_avg_cost}
Let $X$ be a sequence of queries, and let $T_0$ be an initial tree.
Let $A$ be an algorithm that serves $X$ and let $T_t$ be the tree that $A$ has after serving $x_t$. Let $P_t$ be the set of nodes on the path from the root to $x_t$ in $T_{t-1}$ and let $U_t$ be the set of nodes of the minimal subtree that contains all the edges that were rotated by $A$ to transform $T_{t-1}$ to $T_t$. Then the cost of $A$ for serving $X$ at time $t$ is $|P_t \cup U_t|$, and the cost of $A$ for serving $X$ is the sum of costs over $t=1,\ldots,m$. We denote the cost of $A$ to serve $X$ starting with $T_0$ by $cost(A,X,T_0)$. We denote the average cost per query by $\hat{c}(A,X,T_0) = \frac{cost(A,X,T_0)}{m}$. When $T_0$ is clear from the context, or immaterial, we write $cost(A,X)$ and $\hat{c}(A,X)$.
\end{definition}

\begin{definition}[Depth]
\label{definition_depth_in_tree}
Let $T$ be a tree.
The depth of a node
 $v \in T$, denoted by $d(v)$,
 is the number of edges in the path from the root to $v$ (in particular $d(root) = 0$). Note that the cost of querying $v$ (without restructuring) is $d(v) + 1$. We also define the depth of the tree, denoted by $d(T)$, as the maximum depth of a node in $T$, that is $d(T) = \max_{v \in T}{d(v)}$.
\end{definition}

\begin{definition}[Competitiveness]
\label{definition_competitive_ration}
We say that an algorithm $A$ is $(\alpha,\beta)$-competitive for initial tree $T_0$ if for any sequence of queries $X$, it holds that $cost(A,X,T_0) \le \alpha \cdot cost(OPT,X,T_0) + \beta$ where $OPT$ is a best algorithm to serve $X$ given $T_0$ (with full knowledge of $X$). When we do not specify $T_0$ we mean that the relation holds for all initial trees. We refer to $\alpha$ as the \emph{multiplicative term} and to $\beta$ as the \emph{additive term}. For ease of language, we regard the multiplicative term as the \emph{competitive ratio}, and also write ``the competitive ratio of'' instead of ``the multiplicative term of the competitiveness of''. In such cases, we assume that the additive term is $o(m)$. It is easiest to think of $\beta = O(n)$ while assuming that $m = \omega(n)$.
\end{definition}

To conclude this section, we give a precise description of the $GF$ algorithm, in Algorithm~\ref{alg_greedy_future}.
We emphasize that its implementation is complex and probably would not be good in practice. However, its main benefit is its theoretical value, as a candidate for dynamic optimality. Should it be proven to be dynamically-optimal, then we would get a better understanding of the problem and also a stepping-stone to analyze simpler algorithms, such as splay, in comparison to $GF$ rather than against some ``vague'' optimum that depends on the sequence.

\begin{algorithm}[ht]
    \SetAlgoLined
    \DontPrintSemicolon
    \KwIn{
         A sequence of queries $X \in [n]^m$ and an initial BST $T_0$. 
         We restructure $T_{t-1}$ to
         $T_t$ after serving 
         the request $x_t$  with  $T_{t-1}$ for $t=1,\ldots,m$.
    }

    % \KwOut{}
    
    \SetKwFunction{funcInit}{Restructure}
    \SetKwProg{Fn}{Function}{:}{}
    
    \Fn{\funcInit{query value $v$, current tree $T_{t-1}$, future accesses $X'$}}{
        Let $v_1 < v_2 < \ldots < v_k$ be the nodes on the path from the root of $T_{t-1}$ to the queried value $v$ (including $v$ and the root). We also define $v_0 = -\infty$ and $v_{k+1} = +\infty$. Denote the subtrees hanging off this path by $R_0,\ldots,R_k$. \;
        
        For each $i=1,\ldots,k$, let $\tau(v_i)$ be the index of the first appearance of a query of a value $x \in (v_{i-1},v_{i+1})$ in $X'$. Restructure the nodes $v_1,\ldots,v_k$ as a treap: maintain a BST ordering, while the heap's priorities are set to be the $\tau$ values, where the root's $\tau$ is smallest. Tie-break arbitrarily, e.g. in favor of smaller values, or smaller depth prior to restructuring. Then, hang the subtrees $R_0,\ldots,R_k$ unchanged at their appropriate locations. The resulting tree is $T_t$. \;
    }

    \caption{GreedyFuture ($GF$) Algorithm}
    \label{alg_greedy_future}
\end{algorithm}

%%% %%% %%% \input{section_stable_sequences.tex}
\section{Stable Sequences and Lower Bounds}
\label{section_stable_sequences}

In this section we properly define the family of \emph{stable sequences} (Definition~\ref{definition_stable_improved_better} and Definition~\ref{definition_stable_improved}) for which the tree maintained by $GF$ is never changed (i.e.\ the access path of the current query is a treap with respect to the suffix of the sequence). To prove our lower bounds we use such sequences in which only the items at the leaves are requested, and the internal nodes cause some extra cost to $GF$, which $OPT$ avoids. We use a natural way to represent such sequences as trees, and use this representation to prove the lower bounds in this section.

\begin{restatable}{theorem}{theoremMultiplicativeTwoLowerBound}
\label{theorem_multiplicative_2_lower_bound}
If $GF$ is $(c,d)$-competitive where the additive term $d$ is sublinear in the length of the sequence, i.e. $d = o(m)$, then $c \ge 2$.
%%% Restate as "\theoremMultiplicativeTwo*"
\end{restatable}

\begin{restatable}{theorem}{theoremAdditiveLoglogn}
\label{theorem_additive_loglogn}
For every $n \ge 2$ there exist sequences $X \in [n]^m$ such that $cost(GF,X) = cost(OPT,X) + \Omega(m \cdot \lg \lg n)$.
Among these sequences, there exists a sequence whose length is $m = n^{\Theta(\frac{\lg \lg n}{\lg \lg \lg n})}$ for $n \ge 5$, and $m = \Theta(n)$ for $n = O(1)$. (There exist other longer sequences too.)
%%% Restate as "\theoremAdditiveLoglogn*"
\end{restatable}

Theorem~\ref{theorem_multiplicative_2_lower_bound} enables us to prove the following two theorems, proven in Appendix~\ref{appendix_section_missing_proofs}.

\begin{restatable}{theorem}{theoremGFOnSubsequence}
\label{theorem_GF_on_subsequence}
For any $\epsilon > 0$ there exists a sequence $X$ with a subsequence (not necessarily consecutive) $X' \subseteq X$ such that $cost(GF,X') \ge (2-\epsilon) \cdot cost(GF,X)$.
%%% Restate as "\theoremGFOnSubsequence*"
\end{restatable}

\begin{restatable}{theorem}{theoremGFOnReverseSequence}
\label{theorem_GF_on_reverse_sequence}
Let $S$ be a sequence, we define $rev(S)$ to be the sequence $S$ in reverse. For any $\epsilon > 0$ there exists a sequence $X$ such that $cost(GF,rev(X)) \ge (2-\epsilon) \cdot cost(GF,X)$. 
%%% Restate as "\theoremGFOnReverseSequence*"
\end{restatable}

The motivation for studying subsequences (Theorem~\ref{theorem_GF_on_subsequence}) is the  fact that $OPT$ always saves costs when queries are removed from its sequence. Formally, if $X' \subseteq X$, then $cost(OPT,X') \le cost(OPT,X)$. Indeed, $OPT$ can serve $X'$ by simulating a run on  $X$. More generally, this relation of costs when comparing a sequence to a subsequence of it, is an important property which even has a name:

\begin{definition}[Approximate-monotonicity~\cite{DionHarmon2006Thesis,LevyTarjan2019Thesis}]
\label{definition_approximate_monotone}
An algorithm $A$ is approximately-monotone with a constant $c$ if for any sequence $X$, subsequence $X' \subseteq X$, and initial tree $T$, it holds that $cost(A,X',T) \le c \cdot cost(A,X,T)$.
\end{definition}

\begin{corollary}
If $GF$ is approximately-monotone with a constant $c$, then  $c\ge 2$.
\end{corollary}

As noted, $OPT$ is approximately-monotone with $c = 1$ (strictly monotone). The reason that approximate-monotonicity is of interest, in particular for $GF$, is because it is one of two properties that together are necessary and sufficient for any dynamically-optimal algorithm.  The complementing property, which $GF$ is known to satisfy, is simulation-embedding:

\begin{definition}[Simulation-Embedding~\cite{LevyTarjan2019Thesis}]
\label{definition_simulation_embedding}
An algorithm $A$ has the simulation-embedding property with a constant $c$ if for any algorithm $B$ and any sequence $X$, there exists a supersequence $Y \supseteq X$ such that $cost(A,Y) \le c \cdot cost(B,X)$. ($X$ is a subsequence of $Y$, not necessarily of consecutive queries.)
\end{definition}

An algorithm $A$ which is approximately-monotone with a constant $c_1$ and has the simulation-embedding property with a constant $c_2$ is dynamically-optimal with a constant $c_1 \cdot c_2$. Indeed, for any sequence $X$, there is some supersequence $Y(X) \supseteq X$ such that $cost(A,X) \le c_1 \cdot cost(A,Y(X)) \le c_1 \cdot c_2 \cdot cost(OPT,X)$. Harmon~\cite{DionHarmon2006Thesis} proved that $GG$, and hence $GF$, has the simulation-embedding property, hence $GF$ is dynamically-optimal if and only if it is approximately-monotone. An alternative indirect proof was given by~\cite{BSTInversions2020}, proving that $GG$ is $O(1)$-competitive versus the move-to-root algorithm, therefore inheriting the property from move-to-root.
%%% It can be shown that in the Geometric-model of cost, that there exists $Y\supseteq X$ such that $cost(GG,Y) \le 5 \cdot cost(A,X) - 4m$ (for initial tree of our choice) or $cost(GG,Y) \le 5 \cdot (cost(A,X) + n - 1) - 4m$ (for adversarial initial tree).

The motivation for studying reversal (Theorem~\ref{theorem_GF_on_reverse_sequence}) is that $OPT$ is oblivious to reversing the sequence of queries, up to an additive difference of $n$. Indeed, to serve a sequence $X$ in reverse, we can pay $n$ to restructure the initial tree $T_0$ to the final tree $T_m$, and then ``reverse the arrow of time'': when serving query $x_t$, also modify the tree from $T_{t}$ to $T_{t-1}$ where $T_i$ is the tree that $OPT$ would get by the end of processing the $i$-th query of $X$, in order. This means that any dynamically-optimal algorithm must be able to serve a sequence of requests and its reverse with the same cost up to a constant factor. Theorem~\ref{theorem_GF_on_reverse_sequence} does not disprove dynamic-optimality for $GF$, but gives some insight of how reversal affects $GF$.

\subsection{Maintaining a Static Tree for \emph{GF}}
\label{section_static_structure_of_GF}
In this section we describe the basic ``tool'' which we use to fix a tree structure for $GF$ despite its dynamic nature. That is, we describe a class of sequences which we call \emph{mixed-stable sequences} such that $GF$ never restructures its tree when serving a sequence in this class. For the sake of simplicity, we assume that the initial tree is structured as we need it to be. Appendix~\ref{section_appendix_initial_tree}  explains how to enforce a specific ``initial'' tree given an arbitrary initial tree, and also argues why this minor issue does not affect the competitive ratio of $GF$.

As noted, our objective is to produce a sequence that ``tricks'' $GF$ into having unnecessary nodes in the core of the tree, such that the requested values are only at the leaves. As an example, consider the classic sequence of queries $X = [1,3,1,3,\ldots]$ with an initial tree 
containing $2$ at the root, $1$ as a left child of the root and $3$ as a right child of the root. Because of the alternating pattern, $GF$ never restructures the tree, and the cost per query is $2$ rather than $1.5$ on average (e.g. when $1$ is in the root, and $3$ is its right child).

We begin by defining stability in the intuitive way.\footnote{We define stability slightly differently from \cite{SadehGFlowerboundSTACS2023} because it turns out that the original definition did not capture every possible stable sequence but only the \emph{mixed-stable} sequences, which are only a subset.}

\begin{definition}[Stability]
\label{definition_stable_improved_better}
Let $T$ be a full binary search tree, and let $X$ be a sequence of queries over the items of $T$. We say that a node $v \in T$ is \emph{stable} if throughout serving $X$, $GF$ maintains the subtree of $v$ such that its elements are fixed (none added or removed) and $v$ is always the root. We say that $T$ is stable with respect to $X$ if every node is stable, and refer to $X$ as a \emph{stable} sequence and to $T$ as a \emph{stable} tree.
\end{definition}

\begin{remark}
\label{remark_multiple_sequences_per_tree}
There can be multiple stable sequences for the same tree, which are not equivalent up to symmetry. For example, consider the balanced tree over $n=7$ nodes. The following cycles of queries keep it stable (with appropriate tie-breaking for each): $X_1 = [1,5,3,7,\ldots]$, $X_2 = [1,3,7,5,\ldots]$ and $X_3 = [1,3,7,5,1,7,3,5,\ldots]$. These sequences are not symmetric since $X_1$ alternates between the subtrees of the root on every query, $X_2$ does so only every second query, and $X_3$ alternates $\frac{3}{4}$ of the times.
\end{remark}

\begin{lemma}
\label{lemma_alternation_count_ratio}
Let $T$ be a stable BST with respect to a sequence of queries $X$. Fix a non-leaf node $v$ with child $u < v$, and focus only on the subset of queries $X_v \subseteq X$ that access the subtree of $v$, $T_v$. Two consecutive accesses in $X_v$ to the range $[\min(T_v),v)$ can only be such that the first access is $x_1 \in [\min(T_v),u]$ and the second access is $x_2 \in (u,v)$. In particular, there cannot be more than two consecutive accesses in $X_v$ to the subtree rooted at $u$.

Symmetrically, for a child $u > v$, two consecutive accesses in $X_v$ to the range $(v,\max(T_v)]$ can only be such that the first access is $x_1 \in [u,\max(T_v)]$ and the second access is $x_2 \in (v,u)$.
\end{lemma}

\begin{proof}
We prove only the case $u < v$ since the case $u > v$ is symmetric. Denote $z \equiv \min(T_v)$ and the ranges $A_L \equiv [z,u]$, $A_R \equiv (u,v)$ and $A \equiv A_L \cup A_R = [z,v)$. Also denote by $S(x_i)$ the set of nodes touched by query $x_i$.

Now assume by contradiction that there are two consecutive accesses $x_1,x_2 \in A$ not as stated by the claim. First, we show that $x_1 \in A_R$ cannot be followed by $x_2 \in A_R$. By definition of $GF$, following the access to $x_1$, $GF$ will prioritize and promote the nodes $w_1 \equiv \min \{S(x_1) \cap [x_2, v] \}$ and $w_2 \equiv \max \{ S(x_1) \cap [z,x_2] \}$ in anticipation of the access to $x_2$. Since $x_1 \in A_R$, $w_1 < v$ or $w_2 > u$ (or both), which means that $GF$ would change its tree. (Note that the argument holds also if $w_1=w_2=x_2$.)

Second, we show that $x_1 \in A$ cannot be followed by $x_2 \in A_L$. Again by definition of $GF$, following the access to $x_1$, $GF$ will prioritize and promote $w_1$ and $w_2$ in anticipation of the access to $x_2$. Since $w_2 \le w_1 \le u < v$, this means that $GF$ would change its tree. The two cases together imply that two consecutive accesses to $A$ can only happen when $x_1 \in A_L$, followed by $x_2 \in A_R$.
\end{proof}

Lemma~\ref{lemma_alternation_count_ratio} quantifies and bounds the left/right access ratio for each node $v$ of a stable tree, which is the number of accesses to the left subtree of $v$, divided by the number of accesses to the right subtree of $v$. This ratio ranges from $\frac{1}{2}$ to $2$.\footnote{Technically the ratio is not bounded in $[\frac{1}{2},2]$ since, for example, if there is just a single access, the ratio is either $0$ (accessed right) or $\infty$ (accessed left). However, if the number of accesses is not finite, the ratio is arbitrarily close to the range $[\frac{1}{2},2]$.} We define further two specific types of stable nodes, for the special extremal ratios of $1$ and $2$ (or $\frac{1}{2}$).

\begin{definition}[Stable Nodes and Sequences]
\label{definition_stable_improved}
\looseness=-1
Recall Definition~\ref{definition_stable_improved_better}. We say that an inner node $v$ in $T$ is \emph{strongly-stable} if the access pattern in $v$ always alternates between its subtrees, and hence implies an access ratio of $1$.\footnote{A ratio of $1$ does not imply strong stability. For example, $X_2$ and $X_3$ in Remark~\ref{remark_multiple_sequences_per_tree}.}
We say that an inner node $v$
with a left child $u$ in $T$ is \emph{weakly-stable with a left-bias} if both $v$ and $u$ have two children,
and the access pattern to the left and right subtrees of $v$ has the ratio $2{:}1$ in favor of $u$. By Lemma~\ref{lemma_alternation_count_ratio}, the access pattern must follow a $3$-cycle as follows. First it accesses the left-subtree of $u$, then the right subtree of $u$, and finally the right subtree of $v$. (It is left-biased because $\frac{2}{3}$ of the accesses are to the left of $v$). Symmetrically, we say that $v$ is \emph{weakly-stable with a right-bias} if $v$ has two children, its right child $u$ has two children, and the restriction of $X$ to accesses in the subtree of $v$ repeats a $3$-cycle consisting of an access to the right subtree of $u$, the left subtree of $u$, and the left subtree of $v$. Notice that $u$ is a strongly-stable node by definition, and we refer to it as the \emph{favored-child} of $v$.

We regard the sequence $X$ as being induced by the tree $T$ with stability ``attached'' to its inner nodes. We distinguish three special cases: (1) If all inner nodes are strongly-stable then we refer to $X$ and $T$ as \emph{strongly-stable}; (2) If exactly half of the inner nodes of $T$ are weakly-stable then we refer to $X$ and $T$ as \emph{weakly-stable} (recall that each weakly-stable node has a strongly-stable favored-child); (3) If every node is either weakly-stable or strongly-stable, we refer to $X$ and $T$ as \emph{mixed-stable}. Note that the cases of strongly-stable and weakly-stable are special cases of mixed-stable, which is itself a special case of general stability (when some nodes are neither weakly-stable nor strongly-stable). For example, in Remark~\ref{remark_multiple_sequences_per_tree}, $X_1$ is strongly-stable while $X_2$ and $X_3$ are stable but not mixed-stable. See also Figure~\ref{figure_stability_pattern}.
\end{definition}

\begin{figure}[t]%[!ht]
	\centering
	\includegraphics[width=0.22\textwidth]{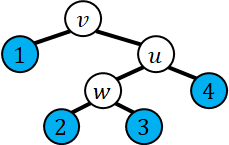}
	
	\caption[Example of node and sequence stability]{\footnotesize{Node and sequence stability (Definition~\ref{definition_stable_improved}). First, consider the repeated sequence $421$, i.e., $X = 421421421 \ldots$. Then $v$ is a weakly-stable right-biased node because its visit pattern is a repetition of $right(u),left(u),left(v)$. $u$ is a strongly-stable node because its visit pattern is $right(u),left(u)$. $w$ is not stable at all, because its visit pattern is always $left(w)$. Second, consider the repetition of the access pattern $12141314$. One can verify that all three inner nodes are strongly-stable. Hence, this is a strongly-stable sequence. Third, no weakly-stable sequence corresponds to the figure, because it requires an even number of inner nodes, but if we make $w$ a leaf (removing $2,3$), then the repeated access pattern of $4w1$ is a weakly-stable sequence.}}
	\label{figure_stability_pattern}
\end{figure}

To motivate Definition~\ref{definition_stable_improved} a little, note that the sequence $X = [1,3,1,3,\ldots]$ is a strongly-stable sequence that corresponds to a tree over the items $\{1,2,3\}$ where $2$ is at the root. $X$ yields a lower bound of $\frac{4}{3}$ on the competitive ratio of $GF$. Similarly, the sequence $X' = [5,3,1,5,3,1,\ldots]$ is a weakly-stable sequence that corresponds to the tree over $\{1,2,3,4,5\}$ with $2$ at the root and $4$ as its right-child. $X'$ yields a lower bound of $\frac{8}{5}$ on the competitive ratio of $GF$, which is already an improvement over the best known lower bound, see also Figure~\ref{figure_basic_stable_trees}. The distinction between strongly-stable and weakly-stable nodes is that $GF$ may modify the structure of the tree when a weakly-stable node is considered, but only temporarily and without affecting the cost. In our example with $X'$, after querying $5$, $GF$ may put $4$ at the root instead of $2$, but following the query to $3$ this change will be reverted. %For the general case, see also the proof of Lemma~\ref{lemma_fixed_structure}, in Appendix~\ref{appendix_section_missing_proofs}.

\begin{figure}[t]%[!ht]
	\centering
	
	\begin{subfigure}[t]{.33\textwidth}
        \centering
        \includegraphics[width=\textwidth]{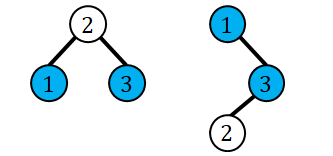}
        \caption{$X = [1,3,1,3,\ldots]$}
        \label{tree4over3} 
    \end{subfigure}
    \hspace{10mm} % '\hfill' for maximal distance. too much here...
    \begin{subfigure}[t]{.33\textwidth}
        \centering
        \includegraphics[width=\textwidth]{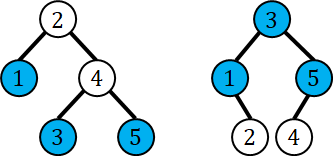}
        \caption{$X' = [5,3,1,5,3,1,\ldots]$}
        \label{tree8over5}
    \end{subfigure}
	
	\caption[The simplest strongly-stable and weakly-stable sequences]{\footnotesize{Examples of the simplest strongly-stable (a) and weakly-stable (b) sequences. Their corresponding trees are the left tree in each pair while the right tree in each pair is an optimized static tree to serve the same sequence. Queried nodes are colored in blue. One can verify that $\hat{c}(GF,X) = 2$ and $\hat{c}(GF,X') = \frac{8}{3}$ while based on the optimized tree, $\hat{c}(OPT,X) \le \frac{3}{2}$ and $\hat{c}(OPT,X') \le \frac{5}{3}$.}}
	\label{figure_basic_stable_trees}
\end{figure}

Motivated by the power of stable sequences over small trees, we proceed to a more general analysis of stable sequences.

\begin{definition}[Atomic Sequence]
\label{definition_notes_on_definition_atomic}
A mixed-stable tree $T$, along with stability type (weak/strong) for each node, and a subtree of each node to be accessed initially, induce a stable sequence. This sequence is unique up to its length, which can be extended indefinitely. We define the ``atomic unit'' of this sequence as the shortest sequence $X$ such that any repetition of $X$ is also a stable sequence that corresponds to $T$.
\end{definition}

Throughout the chapter we work with whole multiples of the atomic sequence. Moreover, unless stated otherwise, we work with the atomic sequence itself (a single repetition). Note that general stable sequences might not have any atomic sequence. This is because there exist infinite stable sequences without a repeating pattern. For example, consider the tree over $n=5$ items with $2$ at the root and $4$ as its right-child. Let $Y_k \equiv [5,3,1] \circ [5,1,3,1]^k$, and define the sequence of queries to be $X = Y_1 \circ Y_2 \circ \ldots \circ Y_k \circ \ldots$. This pattern is not repeating, and one can verify that it is stable. It is not mixed-stable because the root $2$ is neither strongly-stable nor weakly-stable ($4$ is strongly-stable).

\begin{lemma}
\label{lemma_atomic_length}
Let $X$ be a mixed-stable sequence with respect to a tree $T$. Then every leaf $u$ is visited once every $2^{a(u)} \cdot 3^{b(u)}$ queries where $a(u)$ and $b(u)$ are non-negative integers. In particular, the atomic length of $X$ is $2^{\max_{leaf\, u}{a(u)}} \cdot 3^{\max_{leaf\, u}{b(u)}}$ (the lcm). Moreover, if $X$ is strongly-stable then $\forall u: b(u) = 0$, and if $X$ is weakly-stable then $\forall u: a(u) = 0$.
\end{lemma}

\begin{proof}
Consider a leaf $u$. Define the frequency of visiting an ancestor $w$ of $u$ to be the frequency of accessing a leaf in the subtree of $w$. If $w$ is a strongly-stable ancestor then the frequency of visiting a child of $w$ is 
$\frac{1}{2}$ of the frequency of visiting $w$.
If $w$ is weakly-stable, $v$ is its favored-child, and $x$ is a child of $v$ then the frequency of visiting $x$ is $\frac{1}{3}$ of the frequency of visiting $w$. Similarly, if $w$ is weakly-stable and $v$ is its  non-favored child, then the frequency of visiting $v$ is $\frac{1}{3}$ of the frequency of visiting $w$.
It follows that 
 $u$ is visited exactly once every $2^{a(u)} \cdot 3^{b(u)}$ queries where $a(u)$ is the number of strongly-stable nodes that are not favored-children (there are no such nodes if $X$ is weakly-stable), and $b(u)$ is the number of weakly-stable nodes (no such nodes if $X$ is strongly-stable), on the path to $u$. Finally, since every leaf $u$ is visited with a specific period, the whole sequence has a period which is the lcm of all periods.
\end{proof}

\begin{restatable}{lemma}{LemmaFixedStructure}
\label{lemma_fixed_structure}
\looseness=-1
Let $X$ be a mixed-stable sequence with respect to a tree $T$. If $GF$ serves $X$ with $T$ as initial tree, and breaks ties in favor of nodes of smaller depth, then it never restructures $T$.
%%% Restate as "\LemmaFixedStructure*"
\end{restatable}

\begin{proof}
The proof is by induction on the size of the tree. If $T$ has a single node, then it is trivial. Otherwise, the root $r$ is an inner-node, and we prove that it always remains the root. It then follows, by restricting the access sequence to values within each subtree, that the rest of the tree remains fixed as well. We use the notations of $\tau(v)$ and $v_i$ as in Algorithm~\ref{alg_greedy_future}.

First, consider the case that $r$ is a strongly-stable node (Definition~\ref{definition_stable_improved}). Given an access to some value $x$ in the left subtree of $r$, by definition, the next access would be to a value in the right subtree of $r$, hence $\tau(r) < \tau(v_i)$ for any $v_i \ne r$ on the path from $r$ to $x$, and therefore $GF$ will keep $r$ at the root. The same argument holds if $x$ is in the right subtree of $r$, and the next access is in the left subtree.

Next, consider the case that $r$ is a weakly-stable node. By left/right symmetry, without loss of generality assume that it is left-biased, and denote its favored-child (left child) by $u$. Denote the left and right subtrees of $u$ by $A$ and $B$ respectively, and the right subtree of $r$ by $C$. The access pattern of subtrees is $ABC(ABC\ldots)$.
\begin{itemize}
    \item If the current access was to some $x \in A$, both $r$ and $u$ have been touched. The next access queries in $B$, so $\tau(u) = \tau(r) < \tau(v_i)$ for any $v_i \ne u,r$ on the access path to $x$. Since $GF$ tie-breaks in favor of smaller depth, it will keep $r$ at the root.\footnote{This is the reason we defined this kind of access pattern as weakly-stable, because the stability can be chosen, but is not forced. We emphasize that putting $u$ as a parent of $r$ will not make the next access cheaper as both $u$ and $r$ will be touched anyway, and then $r$ will be reinstated as the root.}

    \item If the current access was to some $x \in B$, then both $r$ and $u$ have been touched. The next access touches $C$, so $\tau(r) < \tau(v_i)$ for any $v_i \ne r$ on the access path to $x$, including $u$, thus $r$ must remain the root.

    \item If the current access was to some $x \in C$, since the next access touches $A$, $\tau(r) < \tau(v_i)$ for any $v_i \ne r$ on the access path to $x$, thus $r$ must remain the root. In this case $u$ was not touched, but nonetheless it remains the left child of $r$. \qedhere
\end{itemize}
\end{proof}

\begin{lemma}
\label{lemma_frequency}
Let $X$ be a (general) stable sequence with respect to a tree $T$, which only accesses the leaves of $T$, and let $f(x)$ be the frequency of accessing $x \in X$. Then $f(x) \in [(1/3)^{d(x)} - \epsilon,(2/3)^{d(x)} + \epsilon]$ where $0 \le \epsilon \le \frac{2 \cdot d(x)}{|X|}$. If $X$ is mixed-stable and a whole multiple of an atomic sequence, then $f(x) \in [(1/3)^{d(x)},(1/3)^{d(x)/2}]$, and if $X$ is strongly-stable then $f(x) = (1/2)^{d(x)}$.
\end{lemma}

\begin{proof}
The frequency of visiting a node depends on the path to it. In a general stable sequence, by Lemma~\ref{lemma_alternation_count_ratio} the ratio of accesses is mostly balanced at each ancestor, up to $1:2$. Therefore, ideally we could argue $f(x) \in I$ for $I = [(1/3)^{d(x)},(2/3)^{d(x)}]$. This range should be relaxed because of a possible lead/deficit of an additive $2$ accesses to a subtree versus its sibling, as hinted at the footnote to the paragraph that follows Lemma~\ref{lemma_alternation_count_ratio}. This additive deviation can add up to at most $2 \cdot d(x)$, therefore $f(x) \in [(1/3)^{d(x)} - \epsilon,(2/3)^{d(x)} + \epsilon]$ where $0 \le \epsilon \le \frac{2 \cdot d(x)}{|X|}$.

Focusing on mixed-stable sequences, being a whole multiple of the atomic sequence means that $\epsilon = 0$. The frequency is multiplied by $\frac{1}{2}$ when passing through a strongly-stable node, and multiplied by either $\frac{1}{3}$ or $\frac{2}{3}$ when passing through a weakly-stable node. Every factor of $\frac{2}{3}$ is followed by $\frac{1}{2}$, due to the strongly-stable favored-child of the weakly-stable node. Thus the frequency is bounded between $\frac{1}{3^{d(x)}}$ and $\frac{1}{2^{d(x)/2}} \cdot \big (\frac{2}{3} \big ) ^{d(x)/2} = \frac{1}{3^{d(x)/2}}$.
\end{proof}

\begin{corollary}
\label{corllary_cost_of_GF_formula}
If $X$ is an atomic strongly-stable sequence, $\hat{c}(GF,X) = \sum_{x \in supp(X)}{\frac{d(x) + 1}{2^{d(x)}}}$.
\end{corollary}

\subsection{Promotions and Recursive Trees}
The way in which we show our lower bounds relies on the fact that serving the leaves of a static tree is sub-optimal, since a trivial static optimization is to move the leaves closer to the root. We refer to this operation as a \emph{promotion} of the leaf that we move. We emphasize that for the purpose of our result, we analyze the improvement one gets from promotions, but the actual $OPT$, which is dynamic, may be able to reduce the cost further.

\begin{definition}[Promotion]
\label{definition_promotion}
Consider trees $T$ and $T'$. We say that a node $x$ was promoted  in $T'$ by $h$ (with respect to $T$), if  $d_T(x) - d_{T'}(x) = h$. Given a mixed-stable sequence $X$, the average promotion of $T$ to $T'$ is the weighted average promotion in $T'$ of the nodes of $T$, weighted by the query frequencies of the nodes.
\end{definition}

By definition, static optimization of a tree $T$ to $T'$ for a mixed-stable sequence $X$, implies a cost improvement for $OPT$ which is at least the average promotion of $T$ to $T'$, per query. Intuitively, promoting leaves that are closer to the root contributes more to the average promotion than promoting deeper leaves since the access frequencies decrease exponentially with depth. That being said, our promotion scheme will be relatively uniform, promoting most leaves by roughly the same amount, as in the following example.

\begin{example}
\label{example_delta_cost_1m}
To clarify promotions, consider Figure~\ref{figure_promotion_pattern_for_1_as_chain}. There, we can safely promote every node by one, except for one of the deepest nodes. Therefore, we immediately conclude that for the corresponding strongly-stable sequence $X$, we have:
$\hat{c}(GF,X) \ge \hat{c}(OPT,X) + (1 - \frac{1}{2^n})$.
\end{example}

\begin{figure}[ht]%[!ht]
	\centering
	\begin{subfigure}[t]{.33\textwidth} % Review-version was: \begin{subfigure}[t]{.25\textwidth} 
        \centering
        \includegraphics[width=\textwidth]{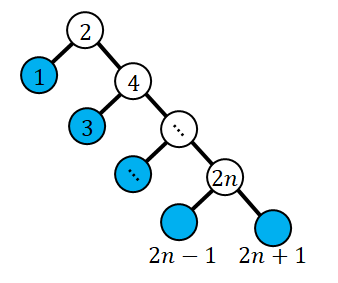}
        \caption{Before promotions.}
        \label{figure_promotion_pattern_for_1_as_chain_before} 
    \end{subfigure}
    \hspace{10mm} % '\hfill' for maximal distance. too much here...
    \begin{subfigure}[t]{.33\textwidth} % Review-version was: \begin{subfigure}[t]{.25\textwidth} 
        \centering
        \includegraphics[width=\textwidth]{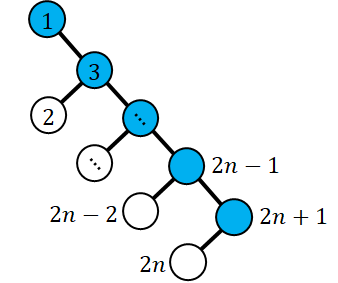}
        \caption{After promotions.}
        \label{figure_promotion_pattern_for_1_as_chain_after}
    \end{subfigure}
    
	\caption{(a) A tree which induces a strongly-stable sequence $X$, only blue nodes are queried. The frequency of querying an odd number $v=2i-1$ in this tree is $\frac{1}{2^i}$ except for $v=2n+1$ which has the same frequency as $v=2n-1$. (b) An improved static tree, in which each node except for one has been promoted one step closer to the root. The cost of serving $X$ over this tree is cheaper by almost $1$ per query.}
	\label{figure_promotion_pattern_for_1_as_chain}
\end{figure}

We define our trees using recursive structures.

\begin{definition}
\label{definition_recursive_structured_trees}
A recursive tree, $T_r$, of depth $r$ is defined by a specific full binary tree $T$ (independent of $r$) such that at least one of its leaves is an actual leaf, and some of its leaves are roots of  recursive trees, $T_{r-1}$, of depth $r-1$. We refer to the inner nodes of $T$ as the \emph{trunk} of $T_r$, and define
%a recursive tree of depth $0$
$T_0$ to be a single node. See Figure~\ref{figure_recursive_tree_pattern} for two examples.\footnote{The name of the pattern $F$ in Figure~\ref{figure_recursive_tree_pattern}, stands for Fibonacci: One can verify that for $r \ge 2$, the number of leaves at depth $1 \le d \le r - 1$ is the $(d-1)$th Fibonnaci number $F_{d-1}$ (we define $F_0 = 0$). Moreover, this can be used to prove the nice equation: $\sum_{d=0}^{\infty}{\frac{F_d}{2^d}} = 2$.}
\end{definition}

\begin{figure}[ht]%[!ht]
	\centering
	\includegraphics[width=0.45\textwidth]{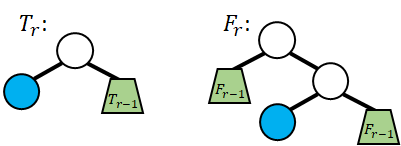} % Review-version was: \includegraphics[width=0.35\textwidth]
	\caption{Two recursive trees of depth $r$. Each of the trees $T$ and $F$ is a full binary tree with at least one actual leaf (in blue), and some hanging subtrees. At the bottom of the recursion (for $r=0$), the subtrees are nodes. Note that: (a) Expanding $T$ for $r=n$ results in the tree in Figure~\ref{figure_promotion_pattern_for_1_as_chain}; (b) The pattern $F$ is important for Theorem~\ref{theorem_multiplicative_2_lower_bound}.}
	\label{figure_recursive_tree_pattern}
\end{figure}

\subsection{Multiplicative Lower Bound for \emph{GF}}
\label{section_multiplicative_lower_bounds}

In this section we prove Theorem~\ref{theorem_multiplicative_2_lower_bound}. We do it by describing a concrete weakly-stable sequence, whose average cost per query is $6$ while an average promotion of $3$ is possible, resulting in an optimal cost of at most $3$. We start by stating a purely mathematical lemma that will be used in the analysis. 
%Its proof is in Appendix~\ref{appendix_section_missing_proofs}.

\begin{restatable}{lemma}{lemmaSemiGeometric}
\label{lemma_geometric_sequence_general}
Let $b_r$ be a sequence
%of numbers
defined by an initial value $b_0$ and the relation $b_r = \alpha \cdot b_{r-1} + \beta + \gamma \cdot \frac{r}{2^r}$ for some constants $\alpha,\beta,\gamma$ where $\alpha \ne \frac{1}{2},1$. Then $b_r =
\frac{\beta}{1-\alpha}(1-\alpha^r) +
\alpha^r \cdot b_0 +
\frac{2\alpha\gamma}{(2\alpha-1)^2} \cdot (\alpha^r - \frac{1}{2^r}) -
\frac{\gamma}{(2\alpha-1)} \cdot \frac{r}{2^r}
$.
In particular, when $\gamma = 0$ then $b_r = \frac{\beta}{1-\alpha} (1 - \alpha^r) + \alpha^r \cdot b_0$.
%%% Restate as "\lemmaSemiGeometric*"
\end{restatable}

\begin{proof}[Proof Sketch]
Either use induction, or ``guess'' that a geometric sequence $y_r$ with a multiplier of $\alpha$ satisfies $y_r = p \cdot \frac{r}{2^r} + q \cdot \frac{1}{2^r} + s + b_r$, and determine the fixed coefficients $p,q,s$.
\end{proof}

\begin{lemma}
\label{lemma_fibo_promotion_weakly_stable}
Let $X$ be a weakly-stable sequence implied by the recursive tree $F_r$ in Figure~\ref{figure_recursive_tree_pattern}, where the root is a weakly-stable node with a right-bias. Then for any $\epsilon > 0$, there is a sufficiently large recursive depth $r$ such that (1) $\hat{c}(GF,X) > 6 - \epsilon$, (2) a static optimization of the tree saves an average cost of at least $3-\epsilon$, and (3) regardless of $r$, $\hat{c}(OPT,X) < 3$.
\end{lemma}

\begin{proof}
Let $c_r$ denote the  average cost of serving $X$ with $F_r$. Then $c_0 = 1$ and $c_r = \frac{1}{3}(c_{r-1}+1) + \frac{1}{3} \cdot 3 + \frac{1}{3}(c_{r-1} + 2) = \frac{2}{3} c_{r-1} + 2$, which yields by Lemma~\ref{lemma_geometric_sequence_general} that $c_r = \frac{2}{1-2/3}(1-(2/3)^r) + (2/3)^r \cdot 1 = 6 \cdot (1-(2/3)^r) + (2/3)^r$. To analyze the average promotion, we re-structure $F_r$ to a new static structure $F'_r$ as follows, see Figure~\ref{figure_promotion_fibo_k2}. The leaf is moved to the root, whose children are the recursive subtrees, optimized themselves by the same logic. The old root is moved to be a right child of the maximal value in the new left subtree, and the old right-child (of the old-root) is moved to be a left child of the minimal value in the new right subtree. $F'_r$ maintains the order of values as was in $F_r$. The demotions of the old root and its right child do not affect the cost, because $X$ does not query these values. Denote by $p_r$ the average promotion of $F_r$  to $F'_r$. Then $p_0 = 0$ since nothing is promoted for a singleton, and $p_r = \frac{1}{3} p_{r-1} + \frac{1}{3} \cdot 2 + \frac{1}{3} (p_{r-1}+1) = \frac{2}{3}p_{r-1} + 1$. Again by Lemma~\ref{lemma_geometric_sequence_general} we get that $p_r = \frac{1}{1-2/3} (1 - (2/3)^r) + (2/3)^r \cdot 0 = 3 \cdot (1 - (2/3)^r)$.
Observe that for $r \to \infty$ we get that $c_r \to 6$ and $p_r \to 3$, thus parts (1) and (2) of the claim follow. For part (3), observe that $c_r - p_r = 6 \cdot (1-(2/3)^r) + (2/3)^r - 3 \cdot (1 - (2/3)^r) = 3 -2 \cdot (2/3)^r < 3$.
\end{proof}

% NOTE: An intuitive and easy way to analyze the average cost of a recursive tree, or the average promotion, is to recall that for $r \to \infty$, $T_r$ and $T_{r-1}$ are ``essentially the same''. This means that instead of solving an exact recursive relation, we can just solve an equation with one variable, $c_\infty$ (same for $p_\infty$) which the limit of $c_r$ (pr $p_r$) must satisfy. For example, a shortcut to proving that $\hat{c}(GF,X) \to 6$ in Lemma~\ref{lemma_fibo_promotion_weakly_stable} would be to solve $c_\infty = \frac{2}{3} c_{\infty} + 2 \Rightarrow c_\infty = 6$.

\begin{figure}[ht]%[!ht]
	\centering
	
	\begin{subfigure}[t]{.42\textwidth} % Review-version was: \begin{subfigure}[t]{.3\textwidth} 
        \centering
        \includegraphics[width=\textwidth]{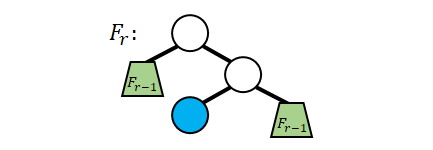}
        \caption{$F$-tree pattern.}
        \label{figure_promotion_fibo_k2_pattern} 
    \end{subfigure}
    %\hspace{10mm}
    \begin{subfigure}[t]{.41\textwidth} % Review-version was: \begin{subfigure}[t]{.3\textwidth} 
       \centering
       \includegraphics[width=\textwidth]{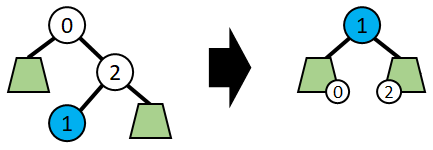}
       \caption{Promotion scheme.}
       \label{figure_promotion_fibo_k2_promotion}
    \end{subfigure}
	
	\caption{The $F$-tree pattern and its promotion scheme in Lemma~\ref{lemma_fibo_promotion_weakly_stable}. Only the top-level promotions are presented in (b), but more promotions are done recursively within each subtree.}
\label{figure_promotion_fibo_k2}
\end{figure}

\theoremMultiplicativeTwoLowerBound*

\begin{proof}
Assume by contradiction that $GF$ is $(2-\delta,f(m))$-competitive for some $\delta>0$ and a function $f(m) = o(m)$. 
Let $X'$ be a sequence that consists of $s$ repetitions of the atomic weakly-stable sequence that corresponds to the recursive tree $F_r$. It follows that $\hat{c}(GF,X') \le (2-\delta) \cdot \hat{c}(OPT,X') + \frac{f(|X'|)}{|X'|}$.
By Lemma~\ref{lemma_fibo_promotion_weakly_stable}, we can choose $r$ large enough such that $\hat{c}(GF,X') > 6 - \delta$, and regardless of $r$, $\hat{c}(OPT,X') < 3$. Then, since $f$ is sub-linear, we can choose the number of repetitions $s$ to be large enough such that $\frac{f(|X'|)}{|X'|} < 2\delta$. But then we also get that $\hat{c}(GF,X') < (2-\delta) \cdot 3 + 2\delta = 6 - \delta$, which is a contradiction.
\end{proof}

By analyzing stable sequences we proved a lower bound of $2$ on the competitive ratio of $GF$. Theorem~\ref{theorem_2_competitive_on_stable_sequences} gives an upper bound.

\begin{restatable}{theorem}{theoremTwoCompetitiveOnStableSequences}
\label{theorem_2_competitive_on_stable_sequences}
Let $X$ be a (general) stable sequence and let $T$ be the tree that corresponds to it. Assume that $X$ only queries leaves of $T$,\footnote{We need this assumption to relate the access cost to Wilber's lower bound; this assumption is satisfied by definition for mixed-stable and strongly-stable sequences.} each leaf at least once. Then $cost(GF,X,T) < c \cdot cost(OPT,X,T)$ for $c=4$. If $X$ is mixed-stable then $c = \frac{5}{2}$.
If $X$ is strongly-stable, then $c = 2$.
%%% Restate as "\theoremTwoCompetitiveOnStableSequences*"
\end{restatable}

We defer the proof of Theorem~\ref{theorem_2_competitive_on_stable_sequences} to Appendix~\ref{appendix_section_missing_proofs}. 
%We note that Theorem~\ref{theorem_multiplicative_2_lower_bound} shows a lower-bound of $2$ while Theorem~\ref{theorem_2_competitive_on_stable_sequences} shows an upper-bound of $\frac{5}{2}$ (on the competitive ratio, for weakly-stable sequences). 
The upper-bound in Theorem~\ref{theorem_2_competitive_on_stable_sequences} is clearly not tight, since in the proof of Theorem~\ref{theorem_2_competitive_on_stable_sequences} we neglected a term using the inequality $\hat{c}(GF,X) \le \frac{2}{\alpha} \cdot \hat{c}(OPT,X) - \frac{1}{\alpha} \big (1 - \frac{n-1}{2m} \big ) < \frac{2}{\alpha} \cdot \hat{c}(OPT,X)$, for a constant $\alpha$. The lack of tightness is more prominent when $\hat{c}(OPT,X)$ is small, like in the sequence studied in Lemma~\ref{lemma_fibo_promotion_weakly_stable} (for Theorem~\ref{theorem_multiplicative_2_lower_bound}). We suspect that the lower bound in Theorem~\ref{theorem_multiplicative_2_lower_bound} is tight, and more strongly, that the $F$-tree pattern is the best pattern to use. This is based on studying several other recursive patterns, including those in Figure~\ref{figure_recursive_tree_pattern} and Figure~\ref{figure_k_r_tree_structure_definition}: None was stronger, and it also seems that patterns with large costs do not ``compensate'' with large enough promotions.

As a closing remark to the multiplicative results, we note that by the static optimality theorem for $GG$ \cite{Fox2011GGAccessLemma}, competitive analysis against a \emph{static} 
algorithm (i.e.\ an algorithm that does not change its initial tree)
 cannot  show a super-constant lower bound.
Concretely, the theorem states that $cost(GF,X) \equiv cost(GG,X) = O(m + \sum_{i=1}^{n}{n_i \lg \frac{m}{n_i}})$
%. This is derived from the %\emph{access lemma} there, 
and one can verify that the actual constants are $5m + 6 \sum_{i=1}^{n}{n_i \lg{\frac{m}{n_i}}}$. This bound can be re-written as $5m + 6m \cdot H_2(X)$ where $H_2(X) = \sum_{i=1}^{n}{\frac{n_i}{m} \lg \frac{m}{n_i}}$ is the base-$2$  entropy of the frequencies of the values in $X$. By~\cite{StaticNearOptimalTreeMehlhorn}, $cost(OPT^s,X) \ge m \cdot \frac{H_2(X)}{\lg 3}$ where $OPT^s$ is the static optimum, and therefore $cost(GF,X) \le (5 + 6 \lg 3) \cdot cost(OPT^s,X)$.
Thus, no static argument can  show a lower bound larger than $\approx 14.51$.

\subsection{Additive Lower Bounds for \emph{GF}}
\label{section_additive_cost}

In this section we move on to analyze the additive gap between $GF$ and $OPT$. For this, we construct and analyze more elaborate patterns of recursively-defined trees, in order to get a large average promotion when optimizing the structure of the trees. The analysis is more involved since we  cannot simply assume that the depth of the recurrence, $r$,  approaches infinity. Here  $n$ is a function of $r$ and the difference of cost can be meaningful in terms of $n$ only if $n$ is finite.

\begin{definition}
\label{definition_tree_pattern}
For $k \ge 2$, and $r \ge 0$ we define a $(k,r)$-tree $T_r$ as follows.   The tree is recursive of depth $r$ (as in Definition~\ref{definition_recursive_structured_trees}), such that its trunk is composed of a root and a left-chain of length $k-1$ that starts in the right-child of the root. The left child of the deepest node of the trunk is an actual leaf, and the rest of the leaves are $T_{r-1}$ subtrees. $T_0$ is a single node. See Figure~\ref{figure_k_r_tree_structure_definition}. When $k$ is clear from the  context, we also refer to the tree as $T_r$.
\end{definition}

Observe that the tree $F_r$ that was used to prove Theorem~\ref{theorem_multiplicative_2_lower_bound} is in fact a $(k,r)$-tree with $k=2$. When we conclude the analysis, we will get the two ends of a ``tradeoff'' such that on the one end we have a relatively high cost ratio, and on the other a relatively high cost difference. Moreover, we will show that the higher the difference of costs on a sequence induced by $(k,r)$-tree, the closer the cost ratio is to $1$ (comparing $GF$ to $OPT$).

\begin{figure}[t]%[!ht]
	\centering
	\begin{subfigure}[t]{.25\textwidth} % Review-version was: \begin{subfigure}[t]{.22\textwidth} 
        \centering
        \includegraphics[width=\textwidth]{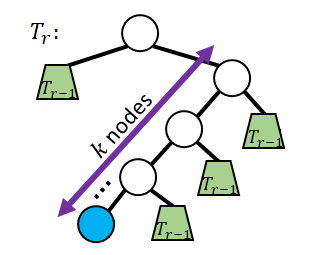}
        \caption{$(k,r)$-tree pattern.}
        \label{figure_k_r_tree_structure_definition_pattern} 
    \end{subfigure}
    \hspace{5mm}
    \begin{subfigure}[t]{.69\textwidth} % Review-version was: \begin{subfigure}[t]{.61\textwidth} 
       \centering
       \includegraphics[width=\textwidth]{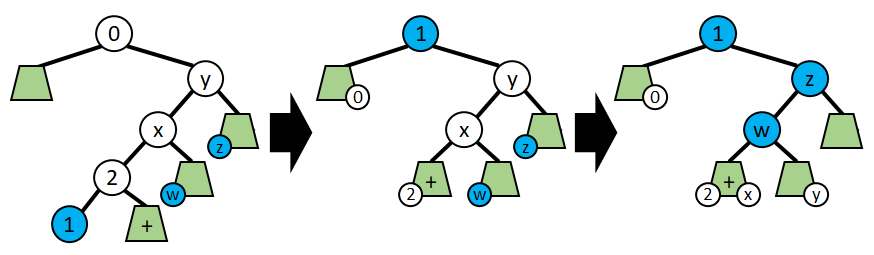}
       \caption{Promotion scheme. The main gain is due to the first step.}
       \label{figure_k_r_tree_structure_definition_promotion}
    \end{subfigure}
    
	\caption{(a) The recursive pattern of a $(k,r)$-tree, $T_r$. The \emph{trunk} of the tree has $k$ nodes: the root, and a chain of $k-1$ nodes leading to an actual leaf. The rest of the leaves are $(k,r-1)$-trees. (b) The promotion scheme used later in Lemma~\ref{lemma_average_promotion}, exemplified for $k=4$ (see also Figure~\ref{figure_promotion_fibo_k2} for the degenerate case of $k=2$). The main gain is from the first step of promoting the actual leaf to the root, and its sibling subtree (marked with $+$) one step upwards. Additional gain is achieved by promoting the left-most node of each hanging right subtree to the trunk at the expense of demoting trunk nodes. More promotions are done recursively within each subtree. The nodes marked $0,1,2$ are indeed consecutive, and also: $2<x$ and $x+1=w<y=z-1$.
	%A $(k,0)$-tree is a single node.
	}
	\label{figure_k_r_tree_structure_definition}
\end{figure}

\begin{lemma}
\label{lemma_deepest_node}
The depth of a $(k,r)$-tree is $k \cdot r$, and its left-most node is at depth $r$.
\end{lemma}
\begin{proof}
Trivial by induction: For $r=0$, the deepest node is the root, at depth $0$. For $r \ge 1$, observe that the deepest node belongs to the deepest subtree $T_{r-1}$, which is rooted at depth $k$ since the path to it includes $k$ trunk nodes. Similarly, the depth of the left-most node is increased by $1$ per recursive level of the tree.
\end{proof}

\begin{lemma}
\label{lemma_size_of_tree}
Let $T_r$ be a $(k,r)$-tree. Then $|T_r| = (2 + \frac{2}{k-1}) k^r - (1 + \frac{2}{k-1})$ where $|T_r|$ is the number of nodes in $T_r$. In rougher terms, $|T_r| = \Theta(k^r)$.
\end{lemma}

\begin{proof}
Denote $n_r = |T_r|$. By definition, $n_0 = 1$ and $n_r = (k+1) + k \cdot n_{r-1}$. Hence, by Lemma~\ref{lemma_geometric_sequence_general} (with $\gamma = 0$): $n_r =  \frac{k+1}{1-k} (1 - k^r) + k^r = (2 + \frac{2}{k-1}) k^r - (1 + \frac{2}{k-1})$.
\end{proof}

\begin{restatable}{lemma}{LemmaAveragePromotion}
\label{lemma_average_promotion}
Let $X$ be any mixed-stable sequence corresponding to a $(k,r)$-tree $T_r$. 
Denote the average weighted promotion possible in $T_r$ by $p_r$, where weighting is according to the frequency of querying each leaf. Then $p_r > k \cdot (1 - \alpha^r)$ for $\alpha = 1 - \frac{1}{3^k}$. In particular, if $X$ is a strongly-stable sequence, then $p_r = (k + 1) \cdot (1 - \alpha^r) + \delta$ for $\alpha = 1 - \frac{1}{2^k}$ and $0 \le \delta < \alpha^r$.
%%% Restate as "\LemmaAveragePromotion*"
\end{restatable}

%Due to space constraints, the proof of Lemma~\ref{lemma_average_promotion} is deferred to Appendix~\ref{appendix_section_missing_proofs}.
 
\begin{proof}
We can promote by $k$ every explicit leaf in every $T_{r'}$ for all recursive levels $1 \le r' \le r$, from its location to the root of $T_{r'}$. Only nodes that are $T_0$ leaves do not contribute an explicit promotion of at least $k$, therefore $p_r > k \cdot (1-f)$ where $f$ is the sum of query-frequencies of all $T_0$ leaves (the inequality is strict due to unaccounted subtree promotions). To conclude, we argue that $f \le (1 - \frac{1}{3^k})^r$. The frequency of accessing the explicit leaf of $T_r$ is at least $\frac{1}{3^k}$ by Lemma~\ref{lemma_frequency}, hence with frequency of at most $1 - \frac{1}{3^k}$ we query a value in some $T_{r-1}$ subtree. Similarly, within the chosen subtree there is again a relative frequency of at most $1 - \frac{1}{3^k}$ to query within some $T_{r-2}$ subtree. Overall, since there are $r$ levels of recursion, we conclude that $f \le (1 - \frac{1}{3^k})^r$.

Proving the second part of the claim required a more careful analysis. We define the following method of promotion, depicted in Figure~\ref{figure_k_r_tree_structure_definition}. In the $(k,r)$-tree we promote the (only) explicit leaf to the root, and promote its sibling subtree by $1$. 
Then we apply similar promotions recursively within every $(k,r-1)$-subtree. Finally, we promote the left-most node within each $(k,r-1)$-subtree that hangs as a right-subtree from the trunk to the parent of this subtree. Denote the total average (weighted) promotion by $p_r$.
Note that it does not matter if we promote the left-most nodes of the right subtrees before or after the recursive promotions, because the total order on the items guarantees that there is only one value that can be put instead of every demoted trunk node, and the recursive promotions within a specific subtree do not change the depth of its leftmost leaf.

The promotion of the explicit leaf of $T_r$ saves a cost of $k$ weighted by a factor (query frequency) of $\frac{1}{2^k}$. The promotion of the sibling subtree saves $1$ weighted by a factor of $\frac{1}{2^k}$. The recursive promotions are $p_{r-1}$ weighted by $\sum_{i=1}^{k}{\frac{1}{2^i}}$ (for all the $k$ subtrees), and finally the last promotions are technically negligible (as seen in the analysis below), but for the sake of completeness we consider them in the analysis as well: promoting the left-most node from each subtree saves $(r-1)+1 = r$ since the leaf that we promote last is at depth $r-1$ within the recursive subtree, and this promotion is weighted by $\frac{1}{2^r} \cdot \sum_{i=2}^{k-1}{\frac{1}{2^i}}$ (factor of $\frac{1}{2^r}$ follows from Lemma~\ref{lemma_deepest_node}). We get that:
$p_r = \frac{k+1}{2^k} + p_{r-1} \cdot \big (1-\frac{1}{2^k} \big ) + \frac{r}{2^r} \cdot \frac{1}{2} \big ( 1 - \frac{1}{2^{k-2}} \big )$.
%%% [I had to make the above inline-equation to save sapce.]
Then by Lemma~\ref{lemma_geometric_sequence_general}, with $\alpha = 1 - \frac{1}{2^k}$ and $\gamma = \frac{1}{2}(1 - \frac{1}{2^{k-2}})$, we get:
$$ p_r = (k+1) \cdot (1-\alpha^r) + \delta \ \ \ , \ \ \ \delta \equiv 
\alpha^r \cdot p_0 +
\frac{2\alpha \gamma}{{(2\alpha -1)}^2} \cdot \Big (\alpha^r  - \frac{1}{2^r}\Big ) - \frac{\gamma}{(2\alpha-1)} \cdot \frac{r}{2^r} $$

It remains to show that $0 \le \delta < \alpha^r$. It is simple to see that $\delta = 0$ for $k=2$, because then $\gamma = 0$ and $p_0 = 0$ is the average weighted promotion in a tree with a single node. For $k \ge 3$, by the definition of $\alpha$ and $\gamma$ we have that
$\frac{2\alpha \gamma}{{(2\alpha -1)}^2} = \frac{(2^k-1)(2^k-4)}{(2^k-2)^2} = 1 - \frac{1}{2^k - 4 + \frac{4}{2^{k}}} \in (\frac{3}{4},1)$ and $\frac{\gamma}{2\alpha - 1} = \frac{1}{2} - \frac{1}{2^k - 2} \in [\frac{1}{3},\frac{1}{2})$. Substituting these bounds 
and $p_0 = 0$ into the formula for $\delta$
 gives $\delta < \alpha^r$. Moreover, $\delta$ is positive since
 $\delta >
 \frac{3}{4} (\alpha^r - \frac{1}{2^r}) - \frac{1}{2} \cdot \frac{r}{2^r} =
 \frac{3}{4} (\alpha - \frac{1}{2}) \cdot \sum_{i=0}^{r-1}{\alpha^i \cdot \big( \frac{1}{2} \big )^{r-1-i}} - \frac{r}{2^{r+1}} =
 \frac{3 (\alpha - \frac{1}{2})}{2^{r+1}} \cdot \sum_{i=0}^{r-1}{(2\alpha)^i} - \frac{r}{2^{r+1}} >
\frac{3(\alpha - \frac{1}{2})}{2^{r+1}} \cdot r - \frac{r}{2^{r+1}}
=
(3 (\frac{1}{2} - \frac{1}{2^k}) - 1)\cdot \frac{r}{2^{r+1}} > 0$ for $k \ge 3$.

Note that indeed the gain from promoting the left-most node of each subtree is negligible, since the effect is merely having $\gamma \ne 0$, which only contributes $0 \le \delta < \alpha^r < 1$.
\end{proof}

\begin{corollary}
\label{corollary_delta_cost}
The average cost-per-query of $GF$ on a strongly-stable sequence induced by a $(k,r)$-tree is larger than the optimal cost by at least $(k+1) \cdot \big (1 - \big(1 - \frac{1}{2^k} \big)^r \big )$. On any mixed-stable sequence, the difference is at least $k \cdot \big (1 - \big(1 - \frac{1}{3^k} \big)^r \big )$.
\end{corollary}

We are ready to prove Theorem~\ref{theorem_additive_loglogn}.

\theoremAdditiveLoglogn*

\begin{proof}
First we prove for $2 \le n = O(1)$. For $n=2$ the claim is trivial since $\lg \lg n = 0$ and any sequence satisfies $cost(GF,X) = cost(OPT,X) + \Omega(0)$. For $n \ge 3$, just consider the sequence in Figure~\ref{tree4over3}, where the cost gap is $\Omega(1)$ per query which is indeed $\Omega(\lg \lg n)$ since $n = \Theta(1)$. We can enforce the desired ``initial'' tree\footnote{$2$ is the root, with children $1$ and $3$, and any remaining items form some right subtree of $3$.} with only $O(1)$ queries by Theorem~\ref{theorem_stabilization_oblivious_better}. Having $m = \Theta(n)$ enables us to amortize the possible initial overhead cost due to items $x \ge 4$ that might be touched once until we enforce the desired tree.

Next we prove for $n \ge 5$, using $(k,r)$-trees. We begin by assuming that $n$ is exactly the size of some $(k,r)$-tree $T_r$ and then explain how to extend to other values in-between.

Let $X$ be the strongly-stable sequence induced by a $(k,r)$-tree $T_r$, and for simplicity assume that the initial tree is $T_r$.\footnote{We remove this assumption in Remark~\ref{remark_amplification_additive_ok}.} By Lemma~\ref{lemma_size_of_tree},
$n = (2 + \frac{2}{k-1}) k^r - (1 + \frac{2}{k-1})$ therefore $\lg \lg n = \lg r + \lg \lg k + O(1)$.\footnote{$k \ge 2 \Rightarrow k^r \le n < 4k^r \Rightarrow \lg n = r \lg k + c$ for $c \in [0,2)$, and so $\lg \lg n = \lg r + \lg \lg k + O(1)$.}
By Corollary~\ref{corollary_delta_cost},
$\hat{c}(GF,X) - \hat{c}(OPT,X) \ge \Delta \equiv (k+1) \cdot (1 - (1 - \frac{1}{2^k})^r)$. By choosing $r = 2^k$ we get that $\Delta = (k+1) \cdot (1 - (1 - \frac{1}{2^k})^{2^k}) \approx (1-\frac{1}{e}) \cdot (k+1)$.\footnote{The approximation is off by less than $10\%$ for $k \ge 2$. ($60\%$ and $20\%$ for $k=0,1$ respectively.)} We also get that $\lg \lg n = k + \lg \lg k + O(1)$, therefore $\Delta \approx (1-\frac{1}{e}) \lg \lg n$ and we conclude that 
$\hat{c}(GF,X) - \hat{c}(OPT,X) \ge \Omega(\lg \lg n)$. 

%As for the length of the sequence, 
By Lemma~\ref{lemma_atomic_length} the length of the atomic strongly-stable sequence of $T_r$ is $m = 2^{d(T_r)}$, hence $m = 2^{rk}$ by Lemma~\ref{lemma_deepest_node}. By Lemma~\ref{lemma_size_of_tree}, $\frac{n + (1+2/(k-1))}{2 + 2/(k-1)} = k^r = 2^{r \lg k}$. Together we get that $m = 2^{rk} = 2^{(r \lg k) \cdot (k/\lg k)} = \Big ( \frac{n + (1+2/(k-1))}{2 + 2/(k-1)} \Big )^{(k/\lg k)} = n^{\Theta(\frac{\lg \lg n}{\lg \lg \lg n})}$.

So far, we proved the claim for an infinite set of values $n \in B$ roughly spaced like $\{ k \ge 2 : 2 \cdot k^{2^k} \}$, but not for every $n \ge 2$. First, we reduce the gaps in the set by allowing to choose $r \in [\frac{2^k}{4},2^k]$ instead of fixing $r = 2^k$, to get $B' \supseteq B$. This may only tweak the lower bound on the gap down to at least $\Delta \approx (k+1) \cdot (1 - \frac{1}{e^{1/4}})$ which is still $\Theta(k+1)$ so the asymptotic claim holds. Now the gaps in the set are much smaller: For the same $k$, $\frac{2^k}{4} \le r < r+1 \le 2^k$ we get $\frac{n(k,r+1)}{n(k,r)} \approx k = \Theta(\log \log n)$, and when increasing $k$ by $1$ we have: $\frac{n(k+1,(2^{k+1})/4)}{n(k,2^k)}
\approx 
(\frac{k+1}{k})^{2^{k}/2} \cdot \frac{1}{k^{2^{k}/2}}
\approx
e^{2^{k}/2k} \cdot \frac{1}{k^{2^{k}/2}}
=
(\frac{e^{1/k}}{k})^{2^{k}/2}$ which is less than $1$ for $k \ge 2$.
Next, given any $n \ge 2$, regard it as if it was the largest $n' \in B'$ such that $n' \le n$. To properly ignore the items $\{n'+1,\ldots,n\}$ we simply never query them including when we enforce an initial tree. This effectively arranges all of them in some BST that is hanging as the right subtree of the item $n'$.\footnote{Indeed, enforcing an initial tree as detailed in Theorem~\ref{theorem_stabilization_oblivious_better} means that after the first query to $n'$, which happens since it is the rightmost leaf of the ``effective tree'', every other item $n'' > n'$ belongs to and remains in the right subtree of $n'$, never to be accessed. In order to immediately dispose of the excess items into their untouched subtree, we can either: (1) query $n'$ once immediately at the beginning of the sequence to achieve this separation (adding $+1$ to the length of the enforcing prefix); or (2) Regard each item $i \le n'$ as if it was $n' + 1 - i$ (left-right reflection) just so naturally $n'$ is queried first in order to even not need an extra query.} We argued that $B'$ is ``dense'' such that $\frac{n}{O(\log \log n)} \le n' \le n$, then $\log \log n' = \Theta(\log \log n)$ and the claim holds for $n$. By Definition~\ref{definition_tree_pattern} we must have $k \ge 2$, and since we choose $r \in [\frac{2^k}{4},2^k]$ then for $k=2$ we have $r \in [1,4]$. So the smallest $(k,r)$-tree we can work with has $n \ge 5$ (Revisit Figure~\ref{tree8over5}), thus the proof holds for $n \ge 5$.
\end{proof}

\begin{remark}
\label{remark_additive_gap_can_be_mixed_stable}
In the proof of Theorem~\ref{theorem_additive_loglogn}, the sequence $X$ does not have to be strongly-stable, and any mixed-stable sequence $X$ induced by a $(k,r)$-tree $T_r$ works as well. Indeed, Corollary~\ref{corollary_delta_cost} guarantees that $\hat{c}(GF,X) - \hat{c}(OPT,X) \ge \Delta$ for $\Delta = k \cdot \big (1 - \big(1 - \frac{1}{3^k} \big)^r \big )$, and then by choosing $r = 3^k$ we  get that $\Delta = \Theta(k)$, and $k = \Theta(\lg \lg n)$, and $m =  2^{\Theta(rk)} = n^{\Theta(\frac{\lg \lg n}{\lg \lg \lg n})}$.
\end{remark}

\begin{remark}
\label{remark_optimizing_additive_gap}
The choice of $r = 2^k$ in the proof of Theorem~\ref{theorem_additive_loglogn} maximizes our lower bound on the additive gap $cost(GF,X) - cost(OPT,X)$ (up to constants) for our $(k,r)$-trees. Indeed, revisiting the proof, we have that $\Delta$ and $n$ are both functions of $k$ and $r$, and we need to choose $r$ and $k$ to maximize $\Delta$ as a function of $n$. Note that $\Delta = O(k)$ regardless of $r$, and $\lg \lg (n) = \lg r + \lg \lg k + O(1)$. To simplify and eliminate a parameter we define $r = 2^k \cdot f(k)$ for some monotone function $f$. Now we get  simplified relations: $\Delta = (k+1) \cdot (1 - (1 - \frac{1}{2^k})^{2^k f(k)}) \approx (k+1) \cdot (1 - e^{-f(k)})$ and $\lg \lg n = k + \lg f(k) + \lg \lg k + O(1)$. Consider the following two cases.
\begin{itemize}
    \item If $f(k) = \Omega(1)$: then $\exists c \in \mathbb{R}$ such that $\forall k \ge 1: \lg f(k) \ge c$, and therefore $\lg \lg n = \Omega(k)$, written differently $k = O(\lg \lg n)$, which yields $\Delta = O(k) = O(\lg \lg n)$.
    \item If $f(k) = o(1)$: Being $o(1)$ means that $\lim_{k \to \infty}{f(k)} = 0$, so for sufficiently large values of $k$ we can use the  approximation $e^x \approx 1+x$ (that holds for small $x$) to get: $\Delta \approx (k+1) \cdot f(k) = \frac{k+1}{1/f(k)}$. If $\frac{1}{f(k)}$ grows faster than $(k+1)$, we get $\Delta = O(1)$ which does not even grow with $n$. Therefore $\frac{1}{f(k)}$ is increasing, but at a sub-linear rate. Recall that $\lg \lg n = k - \lg \frac{1}{f(k)} + \lg k + O(1)$. Since $\frac{1}{f(k)}$ is sub-linear, we get that 
    $k=\Theta(\lg \lg n)$, which yields  $\Delta = O(k) = O(\lg \lg n)$.
\end{itemize}
\end{remark}

%In the proof of Theorem~\ref{theorem_additive_loglogn} we assumed that the initial tree corresponds to the stable sequence that we analyzed. Remark~\ref{remark_amplification_additive} in Appendix~\ref{section_appendix_initial_tree} explains why we can remove this assumption.

\begin{corollary}
\label{corollary_additive_lglgn}
$GF$ is not $(1,O(m))$-competitive. If the multiplicative term is $1$, then the additive term is at least $\Omega(m \cdot \lg \lg n)$.
\end{corollary}

We have yet to analyze the cost of $OPT$ on a strongly-stable sequence $X$ corresponding to a $(k,r)$-tree that produces the gap of $\Omega(m \cdot \lg \lg n)$ in Theorem~\ref{theorem_additive_loglogn}. If the cost was cheap, say linear, we would get a large competitive ratio as well. However, by Theorem~\ref{theorem_2_competitive_on_stable_sequences} we expect a competitive ratio of at most $2$, and therefore we can conclude without further analysis, that $cost(OPT,X) = \Omega(m \cdot \lg \lg n)$. 
In fact, we prove that $cost(OPT,X) = \Theta(m \cdot \frac{\lg n}{\lg \lg \lg n})$.
It follows that the competitive ratio deteriorates when the additive gap increases. %We prove the following lemmas in Appendix~\ref{appendix_section_missing_proofs}.

\begin{restatable}{lemma}{LemmaCostGFonStableTree}
\label{lemma_cost_greedy_future_on_stable_tree}
Define the constants $\alpha \equiv 1 - \frac{1}{2^k}$ and $\beta \equiv \sum_{j=1}^{k}{\frac{j}{2^j}} + \frac{k+1}{2^k}$. Let $X$ be a strongly-stable sequence induced by  a $(k,r)$-tree. Then $\hat{c}(GF,X) = 2^k \cdot \beta \cdot (1 - \alpha^r) + \alpha^r$. In asymptotic terms: $\hat{c}(GF,X) = \Theta(2^k \cdot (1 - \alpha^r))$.
%%% Restate as "\LemmaCostGFonStableTree*"
\end{restatable}

\begin{proof}
We write a recurrence for the average cost, $c_r$, of $GF$  on the strongly-stable sequence induced by $T_r$.  We have $c_0 = 1$, and
$$c_{r+1}
=
\frac{1+k}{2^k} + \sum_{j=1}^{k}{\frac{j+c_r}{2^j}}
=
\Big (1 - \frac{1}{2^k} \Big ) c_r + \sum_{j=1}^{k}{\frac{j}{2^j}} + \frac{1+k}{2^k} \equiv \alpha \cdot c_r + \beta
$$
($\frac{1+k}{2^k}$ is due to the actual leaf, and the summation is the contribution of all the $T_r$ subtrees.) By Lemma~\ref{lemma_geometric_sequence_general} (with $\gamma=0$), $c_r = \frac{\beta}{1-\alpha} (1 - \alpha^r) + \alpha^r \cdot c_0 = 2^k \cdot \beta (1 - \alpha^r) + \alpha^r$. Since $\alpha = 1 - \frac{1}{2^k} \in [\frac{3}{4},1)$ clearly $\alpha^r < 1$.
Furthermore, $\beta = \Theta(1)$.
To see this note that $\beta$ only depends on $k$. Denote $\beta = \beta(k)$ and observe that: $\beta(k+1) - \beta(k) = \big (\frac{k+1}{2^{k+1}} + \frac{k+2}{2^{k+1}} \big ) - \frac{k+1}{2^k} = \frac{1}{2^{k+1}}$.
Therefore, $\beta(k) = \beta(2) + \sum_{i=3}^{k}{\Big( \beta(i) - \beta(i-1) \Big)} = \Big( \frac{1}{2} + \frac{2}{4} + \frac{3}{4} \Big) + \sum_{i=3}^{k}{\frac{1}{2^{i}}} = 2 - \frac{1}{2^k}$, and $\beta(k) \in [\frac{7}{4},2) \Rightarrow \beta = \Theta(1)$. Because $2^k \cdot (1 - \alpha^r) \ge 2^k \cdot (1 - \alpha) = 1 > \alpha^r$, we conclude that $c_{r} = \Theta(2^k \cdot (1 - \alpha^r))$.
\end{proof}

\begin{restatable}{lemma}{LemmaOPTLowerBound}
\label{lemma_opt_lower_bound}
Let $X$ be a sequence from the family of sequences in Theorem~\ref{theorem_additive_loglogn}, then $cost(OPT,X) = \Theta(m \cdot \frac{\lg n}{\lg \lg \lg n})$.
%%% Restate as "\LemmaOPTLowerBound*"
\end{restatable}

\begin{proof}
Let $X$ be a strongly-stable sequence induced by querying a $(k,r)$-tree. We know that $\frac{1}{2} \hat{c}(GF,X) < \hat{c}(OPT,X) \le \hat{c}(GF,X)$ where the lower-bound is by Theorem~\ref{theorem_2_competitive_on_stable_sequences}. Therefore, $\hat{c}(OPT,X) = \Theta(2^k \cdot (1 - \alpha^r))$ by  Lemma~\ref{lemma_cost_greedy_future_on_stable_tree}. By Lemma~\ref{lemma_size_of_tree}, $\lg n = r \cdot \lg k + O(1)$, or $r = \frac{\lg n - O(1)}{\lg k}$. When we substitute $r = 2^k$ as in the proof of Theorem~\ref{theorem_additive_loglogn}, we get that $(1 - \alpha^r) = \Theta(1)$ and $2^k = r = \frac{\lg n - O(1)}{\lg k} = \Theta(\frac{\lg n}{\lg \lg \lg n})$. Therefore, $cost(OPT,X) = \Theta(m \cdot \frac{\lg n}{\lg \lg \lg n})$.
\end{proof}

As a concluding remark, we recall that the $F_r$-tree  is a  $(k,r)$-tree for $k=2$. 
If we substitute $k=2$ in the formula of Lemma~\ref{lemma_cost_greedy_future_on_stable_tree} we get that $\alpha = \frac{3}{4}$, $\beta = \frac{7}{4}$, and  $\hat{c}(GF,X) = 7 \cdot (1 - (3/4)^r) + (3/4)^r$. By Lemma~\ref{lemma_average_promotion}, the average promotion is $3 \cdot (1 - (3/4)^r)$ (for $k=2$, we have $\delta=0$). These values are the strongly-stable analogues of Lemma~\ref{lemma_fibo_promotion_weakly_stable}, and can be used to show a weaker lower bound of $\frac{7}{4}$, on the competitive ratio of $GF$.

%%% %%% %%% \input{section_conclusions.tex}
\section{Conclusions and Open Questions}
\label{section_conclusions}

In this paper we gave improved lower bounds on the competitiveness of the Greedy Future ($GF$) algorithm for serving a sequence of queries by a dynamic binary search tree (BST). In contrast to many of the previous results on $GF$ that are obtained using the geometric-view by studying the equivalent Geometric Greedy ($GG$) algorithm, we used the 
standard ``tree-view'' and the treap-based definition of $GF$. We showed that the competitive ratio of $GF$ is at least $2$, and that there are sequences $X \in [n]^m$ for which the cost difference (additive gap) between $GF$ and $OPT$ is  $\Omega(m \cdot \lg \lg n)$. These lower bounds enabled us  to show that if $GF$ is approximately-monotone (Definition~\ref{definition_approximate_monotone}) with some constant $c$ then $c\ge 2$. Also, the lower bounds show that the cost of $GF$ on a sequence compared to its cost on its reverse,  may differ by a factor as close as we like to $2$. In contrast,
the cost of $OPT$ on a sequence compared to its reverse may differ by at most $n$.

Our results give  new insights on the ``tradeoff'' between the additive term and the multiplicative term in the competitiveness of $GF$, showing that the multiplicative term is typically larger when the total cost of the algorithm on the sequence is smaller. Indeed, our best multiplicative term is achieved for a sequence whose average cost per query is $6$. This tradeoff is not surprising since a fixed difference implies a larger ratio when the quantities are small. It may be interesting to figure out if this tradeoff hints of some underlying property of $GF$, or is just an artifact of our technique that requires high costs on average per query in order to increase the additive gap between $GF$ and $OPT$.

Clearly, these improved lower bounds still don't settle the deeper question of whether $GF$ (and $GG$) is dynamically-optimal. Our techniques focused on a smaller family of sequences which we named \emph{mixed-stable sequences}, whereas ``most'' sequences are not stable.
%Given that we proved a lower bound of $2$, the next non-trivial step would be to improve it (or show it to be the competitive ratio).
While it is possible that an improved lower bound (larger than $2$) can be found by a more clever pattern of mixed-stable sequences, it seems more likely to be found by  analyzing sequences for which the tree maintained by $GF$ is not static. In addition, we note that $GF$ was not investigated too deeply directly, as most of the work has been done in the geometric view with respect to its counterpart $GG$. Therefore, studying other problems in tree-view may give complementing insights. One such problem is the deque conjecture, which has been partially settled for $GG$, in the case when deletions are only allowed on the minimum item~\cite{GeometryExtended-WADS2015}.

\bibliography{reference}

%%% %%% %%% \input{section_appendix.tex}
\label{section_appendix}

\appendix
\section{Appendix: Deferred Proofs and Discussions}

\subsection{Enforcing a Stable Tree for \emph{GF}}
\label{section_appendix_initial_tree}
We describe how to restructure any initial tree, to a desired tree, when $GF$ is considered. The initial tree cannot simply be re-organized since $GF$ updates the tree in a specific way following each query. 
Let $P \circ X$ denote the concatenation of the sequences
$P$  and $X$.
Note that even if $P$ enforces a desired tree when served alone, serving $P \circ X$ may give a different tree following $P$ when $X$ starts. The reason for this is that $GF$ restructures the tree while serving $P$ according to  future queries, therefore the existence of $X$ may affect its decisions while serving $P$. Nevertheless, we present a simple technique to enforce a tree for $GF$.
% in way that requires $O(n^2)$ queries and works for any sequence.

%\label{theorem_stabilization_sequence}

\begin{theorem}
\label{theorem_stabilization_oblivious_better}
For any tree $T$ there is a sequence $S(T)$ such that: (1) $|S(T)| = O(n \cdot d(T))$, and more precisely, $|S(T)| \le \min \Big (3n(d(T) - \floor{\lg n} + 2) , \frac{3}{2}n(n-1) \Big )$, and, (2) for any suffix of queries $Y$, when $GF$ serves $S(T) \circ Y$, its tree when it is done with the last query of $S(T)$ is $T$. We say that $S(T)$ \emph{enforces} $T$.
\end{theorem}

\begin{proof}
We enforce the structure of $T$ bottom-up, by first ensuring the position of the leaves, and continuing recursively upwards towards the root. Define $T^{[0]} \equiv T$ and let $T^{[i+1]}$ be the tree $T^{[i]}$ stripped of all of its leaves, until the final tree $T^{[h]}$ contains only the root. Observe that $h = d(T)$ (the depth of $T$). For each tree $T^{[i]}$ we define $R_i$ to be the set of non-leaf nodes in $T^{[i]}$. Note that a non-leaf node may be binary or unary.

We construct $S(T)$ in steps. In step $i \ge 0$ we query $R_i$ in two monotonic phases, where the first phase queries every item twice, and the second phase queries every item once. Denote the queries of this step by $S_i$. For example, if $R_0 = \{1,3,5\}$ then we query in the first step: $S_0 = [1,1,3,3,5,5,1,3,5]$. Note that $S_h = \emptyset$ since $R_h = \emptyset$ by definition. The resulting sequence $S(T)$ is the concatenation of the queries in all the (non-empty) steps, that is, $S(T) = S_0 \circ S_1 \circ \ldots \circ S_h$.

First we analyze $|S(T)|$. Observe that we strip leaves between steps, $|R_{i+1}| \le |R_i| - 1$. Initially $|R_0| \le n-1$ hence $|R_i| \le n-1-i$. Also, $h < n$. Therefore we get $|S(T)| = \sum_{i=0}^{h}{|S_i|} = \sum_{i=0}^{h}{3|R_i|} \le 3\sum_{i=0}^{n-1}{(n-1-i)} = \frac{3}{2}n(n-1)$.
To get the other bound, observe that $|R_{h-i}| \le 2^i - 1$, which is meaningful for the last few steps, i.e., for $i \le \floor{\lg n}$. With that in mind:
$|S(T)| 
= 3\sum_{i=0}^{h - \floor{\lg n}}{|R_i|} + 3\sum_{i=h - \floor{\lg n} + 1}^{h}{|R_i|}
< 3\sum_{i=0}^{h - \floor{\lg n}}{n} + 3\sum_{j=0}^{\floor{\lg n} - 1}{2^j} < 3n \cdot (h - \floor{\lg n} + 1) + 3n$. In conclusion, $|S(T)| < 3n(d(T) - \floor{\lg n} + 2)$.

Next we prove that the tree of $GF$ is exactly $T$ when it finishes serving $S(T)$ regardless of any suffix of queries. The proof is by induction on $d(T) = h$. The base case is for $h=0$. In this case $S(T) = \emptyset$ (because $S_0 = \emptyset$), while also $T^{[h]} = T^{[0]} = T$ contains only the root. Querying nothing is not a problem since there is a unique tree with a single node, and we are done.
Now, assume that the claim holds for $h=k$. Consider a tree $T$ of depth $k+1$. First we show that when $GF$ finishes processing $S_0$, all the leaves of $T$ are leaves of the tree of $GF$, and will never be touched (accessed or otherwise), therefore they remain leaves until $GF$ finishes processing $S(T)$. See Figure~\ref{figure_enforcing_a_tree} for a visualization. Indeed, when we query a value $u$ twice in a row, $GF$ brings $u$ to the root when reordering the tree after the first query of the couple. Note that since we query $R_0$ monotonically, by the end of the first phase of $S_0$ (the queries in pairs), $GF$ has a tree whose left-spine is exactly the items of $R_0$. Indeed, each item, in turn, is brought to the root and demotes the previous root to the left. Moreover, no value of $T^{[0]} \setminus R_0$ is part of this spine, because of the second phase of queries (monotonically querying the values of $R_0$). To see this, note that if during the first phase we touch $v \notin R_0$ when $r$ is the root and the pair of queries is to $u$, such that $r<v<u$, and $u$ is the successor of $r$ in $R_0$, then the next access to $r$ is closer in the future than that of $v$, so $r$ will indeed be placed as the left-child of the new root $u$ (and $v$ will be the right child of $r$). If $r<u<v$ then $v$ does not interfere with $r$ being the left child of $u$. %since it must remain in the right subtree of $u$.

Now that we know that all the leaves of $T = T^{[0]}$ are fixed as leaves of the tree of $GF$ when it finishes processing $S_0$, we can conclude by induction: $S' \equiv S_1 \circ \ldots \circ S_{k+1}$ is exactly the sequence that enforces $T^{[1]}$, and by our inductive assumption, $T^{[1]}$ is enforced correctly. Since the leaves of $T^{[0]}$ are not touched at all during $S'$, we conclude that they must remain hanging off $T^{[1]}$ as ``subtrees''. The location of each leaf is uniquely determined, and thus we conclude that $S(T)$ indeed enforces $T = T^{[0]}$.
\end{proof}

\begin{figure}[ht]%[!ht]
	\centering
	\begin{subfigure}[t]{.27\textwidth}
        \centering
        \includegraphics[width=\textwidth]{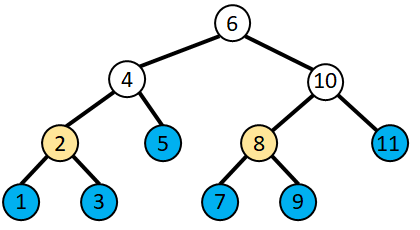}
        \caption{Desired tree $T$.}
        \label{figure_enforcing_0} 
    \end{subfigure}
    \hspace{5mm}
    \begin{subfigure}[t]{.54\textwidth}
       \centering
       \includegraphics[width=\textwidth]{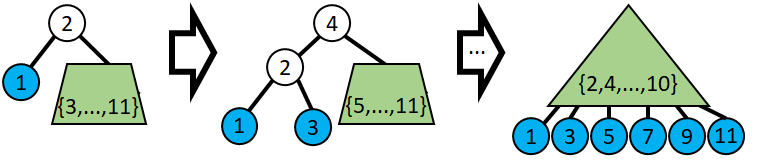}
       \caption{First step of enforcing $T$: fix the leaves.}
       \label{figure_enforcing_step1}
    \end{subfigure}
    \hspace{5mm}
    \begin{subfigure}[t]{.65\textwidth}
       \centering
       \includegraphics[width=\textwidth]{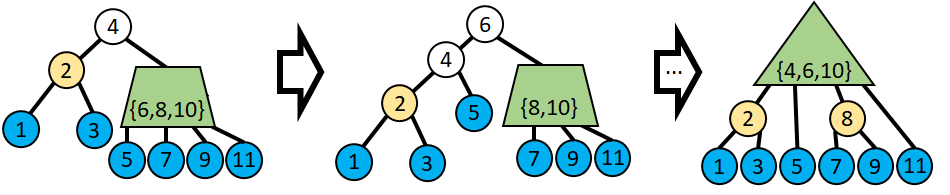}
       \caption{Second step of enforcing $T$: fix the leaves of $T^{[1]}$.}
       \label{figure_enforcing_step2}
    \end{subfigure}
 
	\caption{Example of enforcing a tree as detailed in Theorem~\ref{theorem_stabilization_oblivious_better}. The desired tree $T$ is shown in (a). Blue nodes are the leaves of $T^{[0]} = T$ and golden nodes are the leaves of $T^{[1]}$ (leaves of $T^{[2]}$ and $T^{[3]}$ are not colored). The first step of queries is $S_0 = [2,2,4,4,6,6,8,8,10,10,2,4,6,8,10]$; (b) shows the states after the first query of $2$, the first query of $4$, and by the end of $S_0$. Following the first step, it remains to enforce the remainder of the tree (the inductive step). Concretely, the second step of queries is $S_1 = [4,4,6,6,10,10,4,6,10]$; (c) shows the states after the first query of $4$, the first query of $6$, and by the end of $S_1$. The green trapezoids and triangles abstract away the structure of some subtrees. The last non-empty step is $S_2 = [6,6,6]$ (not shown), after which we get $T$.}
	\label{figure_enforcing_a_tree}
\end{figure}

\looseness=-1
Adding a prefix to our sequence may affect the competitive ratio. However, once we fixed the stable tree, we can repeat  the corresponding stable sequence to ``amplify'' the original competitive ratio making the effect of the prefix negligible. 
%%% [We do not elaborate further due to lack of space. We had a formal lemma to show this amplification-claim...]
One difficulty raised by repetitions is when we care about the length of the sequence in our claim. This is the case in Theorem~\ref{theorem_additive_loglogn} where we claim the existence of a sequence of length $n^{\Theta(\frac{\lg \lg n}{\lg \lg \lg n})}$. In the proof of this theorem we assumed for simplicity that we can choose the initial tree. The following remark shows that indeed we can start with an arbitrary initial tree without weakening the theorem.

\begin{remark}
\label{remark_amplification_additive_ok}
%%% \looseness=-1
Let $X$ be an atomic mixed-stable sequence used to prove Theorem~\ref{theorem_additive_loglogn}, and use the same notation $n'$ for the actual subset of queried items where $\frac{n}{O(\log \log n)} \le n' \le n$. Consider the sequence $Z=S \circ {X^{n'}}$, where $S$ is the prefix (guaranteed by Theorem~\ref{theorem_stabilization_oblivious_better}) that is enforcing the desired ``initial'' tree $T$ that corresponds to $X$, $X^{n'}$ are $n'$ repetitions of $X$, and $\circ$ represents concatenation. We abuse notation and use $T$ for the tree over $n'$ items, and the remaining $n-n'$ items hang ignored to the right of the rightmost tree of $T$. They may only affect the initial cost when converting the initial tree into $T$. There are no unary nodes in $T$ and $X$ queries every leaf at least once so $\frac{n'}{2} < \frac{n'+1}{2} \le |X|$. Together with Theorem~\ref{theorem_stabilization_oblivious_better} we get that $n'|X| \le |Z| < n'(|X|+\frac{3n'}{2}) < 4n'|X|$, therefore we have: $|Z| = \Theta(n'|X|) = n^{\Theta(\frac{\lg \lg n}{\lg \lg \lg n})}$ (the second equality is by Theorem~\ref{theorem_additive_loglogn}). Let $T'$ be the promoted tree we used to prove the cost gap. $OPT$ is at least as good as following $GF$ on $S$, then converting $T$ to $T'$ in $O(n)$ cost, and then serving statically on $T'$, and
since after processing $S$ the tree of $GF$ is fixed: $cost(GF,Z,T_0) - cost(OPT,Z,T_0) \ge n' \cdot (cost(GF,X,T) - cost(STATIC,X,T')) - O(n)$. By the proof of Theorem~\ref{theorem_additive_loglogn} and Remark~\ref{remark_additive_gap_can_be_mixed_stable}: $cost(GF,X,T) - cost(STATIC,X,T') = \Omega(|X| \cdot \lg \lg n)$, and putting everything together we get: $cost(GF,Z,T_0) - cost(OPT,Z,T_0) = \Omega(|Z| \cdot \lg \lg n)$. Note that if $|X| = \Omega((n')^2)$ then it suffices to define $Z = S \circ X$ without repetitions and we get that $|Z| = \Theta(|X|)$, and the rest of the arguments remain the same. There may be additional cases where repetitions are not required, e.g., when the depth of $T$ is $\lg n + O(1)$ (in this case, $|S| = O(n)$).
\end{remark}

\subsection{Omitted Proofs}
\label{appendix_section_missing_proofs}

In this subsection we restate and prove Lemmas and Theorems that were omitted from the main text. 
For convenience, we restate the claims
in their original numbering.

The proof of Theorem~\ref{theorem_2_competitive_on_stable_sequences} makes use of Wilber's first bound~\cite{WilberBounds1989}. We use the original presentation of this bound which is a bit tighter than later simplified versions such as~\cite{KozmaThesis}.

\begin{definition}[Wilber's First Bound~\cite{WilberBounds1989}]
\label{definition_wilber_first_bound}
Let $X$ be a sequence of queries, and let $T$ be a static reference tree such that every query of $X$ is in a leaf of $T$. An alternation at an inner node $u$ of $T$ is defined to be two queries closest in time such that one accesses either the left or right subtree of $u$ and the other accesses the other subtree of $u$. Define $ALT(u)$ to be the number of alternations at node $u$.
%%% [SAVING SPACE] Wilber's first lower-bound for an algorithm $A$ that always moves the queried item to the root is $cost(A,X) \ge WB(X)$ where $WB(X) \equiv m + \sum_{\text{inner\ } u \in T}{ALT(u)}$. $OPT$ can be simulated with an algorithm $OPT'$ that always moves the queried item to the root, and  reverts this move when serving the next request if $OPT$ did not move the item to the root. This gives $WB(X) \le cost(OPT',X) \le 2 \cdot cost(OPT,X) - m$ where the $(-m)$ is because the undo and next access overlap at least at the root of the tree. Therefore,
Then: $cost(OPT,X) \ge m + \frac{1}{2} \sum_{\text{inner\ } u \in T}{ALT(u)}$.
\end{definition}

\theoremTwoCompetitiveOnStableSequences*

\begin{proof}
We use the tree that corresponds to the mixed-stable sequence as the reference tree for Wilber's first bound. Arithmetic manipulations will yield an expression that we can tie to the cost of $GF$, according to the claim.

Let $X$ be a mixed-stable sequence, with a corresponding tree $T$. Let $S$ be the set of values that are in the leaves of $T$, and let $U$ be the set of inner nodes, $|U| = \frac{n-1}{2}$. We also denote by $A(i)$ the set of proper ancestors of $i$. By the definition of the cost of a static tree, we know that $\hat{c}(GF,X) = \sum_{i \in S}{(d(i)+1) \cdot f(i)}$ where $d(i)$ is the depth of $i$ and $f(i)$ is the frequency of accessing $i$. We extend $f(u)$ to refer to the frequency of visiting any node $u$. Note that $f(u) = \sum_{i \in S \wedge u \in A(i)}{f(i)}$ and that $\sum_{i \in S}{f(i)} = 1$.

Now consider Wilber's bound for $X$, with $T$ as the reference tree. We can use $T$ as the reference tree since $X$ only accesses leaves of $T$, by definition. We also denote $\alpha_u \equiv \frac{ALT(u) + 1}{f(u) \cdot m}$ ($ALT(u)$ is defined in Definition~\ref{definition_wilber_first_bound}, and note that $0 \le ALT(u) \le f(u) \cdot m - 1$). We have $\alpha_u \in (0,1]$, where $\alpha_u = 1$ corresponds to fully alternating accesses to the subtree rooted at $u$.
%%% Redacted due to lack of space: \footnote{$ALT(u) \le f(u) \cdot m - 1$ because the number of alternation is at most one less than the number of accesses, due to counting pairs of consecutive accesses.}
The lower bound is $cost(OPT,X) \ge m + \frac{1}{2} \sum_{u \in U}{(ALT(u) + 1)} - \frac{|U|}{2}
= \frac{m}{2} + \frac{m}{2} (1 + \sum_{u \in U}{\alpha_u \cdot f(u)}) - \frac{n-1}{4}
= \big ( \frac{m}{2} - \frac{n-1}{4} \big ) + \frac{m}{2} \sum_{i \in S}{(1 + \sum_{u \in A(i)}{\alpha_u}) f(i)}$. 
Let $\alpha \le \min_{u \in U}{\alpha_u}$, we get that $\hat{c}(OPT,X) \ge \big (\frac{1}{2} - \frac{n-1}{4m} \big ) + \frac{\alpha}{2} \sum_{i \in S}{(d(i)+1) \cdot f(i)} = \frac{\alpha}{2} \hat{c}(GF,X) + \big (\frac{1}{2} - \frac{n-1}{4m} \big )$ where the equality holds since $GF$ maintains a static tree. Thus $\hat{c}(GF,X) \le \frac{2}{\alpha} \cdot \hat{c}(OPT,X) - \frac{1}{\alpha} \big (1 - \frac{n-1}{2m} \big ) < \frac{2}{\alpha} \cdot \hat{c}(OPT,X)$. The last inequality holds since each leaf is queried at least once and $T$ is full, thus $m \ge \frac{n+1}{2} > \frac{n-1}{2}$. Querying every leaf at least once also prevents the degenerate case where some $u$ has $\alpha_u = 0$, so we can safely choose $\alpha > 0$.

In order to choose a suitable $\alpha$, recall that a strongly-stable node $u$ has a coefficient of $\alpha_u = 1$, which means that for strongly-stable sequences, in which all inner nodes are stable, we can pick $\alpha=1$ and conclude that $\hat{c}(GF,X) < 2 \cdot \hat{c}(OPT,X)$. If $u$ is a weakly-stable node, then its coefficient is $\alpha_u = \frac{2}{3}$.
% because there is a bias of $1:2$ ratio of visiting its different subtrees.
So for a mixed-stable sequence we can naively pick $\alpha = \frac{2}{3}$, resulting in $\hat{c}(GF,X) < 3 \cdot \hat{c}(OPT,X)$. For a (general) stable sequence, by Lemma~\ref{lemma_alternation_count_ratio}, $\alpha_u \ge \frac{1}{2}$ so we can choose $\alpha = \frac{1}{2}$ and get $\hat{c}(GF,X) < 4 \cdot \hat{c}(OPT,X)$.

In order to improve for mixed-stable sequences from $3$ to $\frac{5}{2}$, we observe that by definition, every weakly-stable node has a strongly-stable child. Let $u$ be a weakly-stable node and let $w$ be its (strongly-stable) favored-child (recall Definition~\ref{definition_stable_improved}). Since $ALT(u)=ALT(w)$ (by definition of the access pattern in $u$), we can present Wilber's bound differently, summing $(ALT(u)+1) \cdot (1 + \beta) + (ALT(w)+1) \cdot (1 - \beta)$ instead of $(ALT(u)+1) + (ALT(w)+1)$. We get modified coefficients $\alpha'_u = \frac{(ALT(u)+1) \cdot (1 + \beta)}{m \cdot f(u)} = \alpha_u \cdot (1+\beta) = \frac{2(1+\beta)}{3}$ and similarly $\alpha'_w = \alpha_w (1-\beta) = (1-\beta)$. Choosing $\beta = \frac{1}{5}$ balances the coefficients: $\alpha'_u = \alpha'_w = \frac{4}{5}$. Now we can choose $\alpha = \frac{4}{5}$, and get $\hat{c}(GF,X) < \frac{5}{2} \cdot \hat{c}(OPT,X)$ for mixed-stable sequences.
%%% [Note: we would have liked to change $X$ to $X'$ such that $X'$ has more alternations, to get higher $\alpha_u$ values. It is possible to find an equivalent sequence (as formally described here) such that $\forall u: \alpha_u \ge \frac{2}{3}$. Alas, changing $X \to X'$ means that $OPT$ might change, so it does not lead to an improved upper bound, unlike what we did for mixed-sequences. DETAILS FOLLOW.] Given a general sequence $X$ we can reorder its queries slightly to get $X'$ such that no node $u$ has an alternation pattern $LLRR$ or $RRLL$ (where $L$ and $R$ denote accessing its left and right subtree). Instead, we cross the two middle queries to get the respective patterns $LRLR$ or $RLRL$. This works because any other node is oblivious to the change: for higher nodes, the same subtree is being accessed and for lower nodes, one of the subtrees is ``invisible''. $u$ itself is unaffected because alternations only ensure that $GF$ keeps $u$ as the root of its own subtree. We may need to apply this transformation iteratively, for example top down, until $X'$ is fixed for every node $u$. Since there are no two non-alternations at distance $2$ from each other in $X'$, $\alpha_u \ge \frac{2}{3}$. Then choosing $\alpha = \frac{2}{3}$ for $X'$ works, and since the stable tree of $X$ and $X'$ is the same as well the frequency of each item, $\hat{c}(GF,X) = \hat{c}(GF,X') < 3 \cdot \hat{c}(OPT,X')$. Note that here we get stuck, because we related $GF$ on $X/X'$ to OPT only on $X'$ and not on $X$.
\end{proof}

\theoremGFOnSubsequence*

\begin{proof}
Denote the initial tree by $T_0$. Let $Z$ be the weakly-stable sequence used for proving Theorem~\ref{theorem_multiplicative_2_lower_bound}. Let $T_P$ be the tree that corresponds to $Z$ and $T_Q$ the optimized tree, in which the leaves are promoted as in Lemma~\ref{lemma_fibo_promotion_weakly_stable}. Let $P$ and $Q$ be the sequences that enforce $T_P$ and $T_Q$ by Theorem~\ref{theorem_stabilization_oblivious_better}, respectively. Note that $\epsilon$ determines $Z$, $P$ and $Q$ since it tells us how close to a ratio of $2$ we need to get.

Revisit Figure~\ref{figure_promotion_fibo_k2} to see the (recursive) structures of $T_P$ (on the left) and $T_Q$ (on the right, post-promotions). Observe that $T_P$ remains static when $GF$ serves $Z$ with it, by definition. Moreover, $T_Q$ remains static when $GF$ serves $Z$ with it. Indeed, let $r$ be the root of $T_Q$. $Z$ queries the item in $r$ every third access and the other accesses are alternating between its left and right subtrees, hence $r$ remains the root of $T_Q$. The rest of $T_Q$ remains static recursively.

\looseness=-1
Define $X = P \circ Q \circ Z^k$ for a large $k$, and $X' = P\circ Z^k \subset X$ ($\circ$ for concatenation). Since $GF$ does not change $T_P$ and $T_Q$ while serving $Z$ we get that
$\frac{cost(GF,X')}{cost(GF,X)}
= \frac{cost(GF,P, T_0) + k\cdot cost(GF,Z,T_P)}{cost(GF,P \circ Q,T_0) + k\cdot cost(GF,Z,T_Q)}
$. This ratio approaches $\frac{cost(GF,Z,T_P)}{cost(GF,Z,T_Q)}$ for large enough $k$, and since $T_P$ and $T_Q$ are exactly the trees used in the proof of Lemma~\ref{lemma_fibo_promotion_weakly_stable}, we conclude that we can make the resulting ratio as close to $2$ as we like (choosing $Z,P,Q$ according to the desired $\epsilon$).
\end{proof}

\theoremGFOnReverseSequence*

\begin{proof}
The proof is similar to that of Theorem~\ref{theorem_GF_on_subsequence}, and we define $Z$, $T_0$, $P$, $T_P$, $Q$ and $T_Q$ the same way. Here we define $X = Q \circ (rev(Z))^{k+1} \circ rev(P)$, for a large $k$.

We claim that $T_Q$ remains static when $GF$ serves $rev(Z)$ over it, rather than $Z$, by the same argument as in the proof of Theorem~\ref{theorem_GF_on_subsequence}, because the interleaving pattern in the root is preserved under reversal. Moreover, $cost(GF,rev(Z),T_Q) = cost(GF,Z,T_Q)$ because the cost on a static tree depends only on the access frequencies. Putting everything together, we get: $\frac{cost(GF,rev(X))}{cost(GF,X)} = \frac{cost(GF,P,T_0) + k \cdot cost(GF,Z,T_P) + cost(GF,Z \circ  rev(Q),T_P)}{cost(GF,Q,T_0) + k \cdot cost(GF,rev(Z),T_Q) + cost(GF,rev(Z) \circ rev(P),T_Q)}$. Note that the suffix contains one repetition of $Z$ so that the rest of it ($rev(P)$ or $rev(Q)$) does not affect the restructuring decisions of $GF$ during the earlier repetitions of $Z$. The limit of this ratio for large $k$ is $\frac{cost(GF,Z,T_P)}{cost(GF,Z,T_Q)}$. We finish the argument as in the proof of Theorem~\ref{theorem_GF_on_subsequence}.
\end{proof}

\end{document}